\documentclass[10pt, draftclsnofoot, onecolumn]{IEEEtran}

\usepackage{amsmath}
\usepackage{amssymb}
\usepackage{amsthm}
\usepackage{cite}

\usepackage{graphicx}
\usepackage{enumerate}
\usepackage{setspace}
\usepackage{subfigure}
\usepackage{color}
\usepackage{multirow}
\usepackage{slashbox}

\newtheorem{theorem}{Theorem}

\newtheorem{lemma}{Lemma}

\title{The GDOF of 3-user MIMO Gaussian interference channel}

\author{Jung Hyun Bae, Jungwon Lee, Inyup Kang\\
Mobile Solutions Lab\\
Samsung US R$\&$D Center\\
San Diego, CA, USA\\
Email: jungbae@umich.edu, jungwon@alumni.stanford.edu, inyup.kang@samsung.com}

\begin{document}
\maketitle
\begin{abstract}
The paper establishes the optimal generalized degrees of freedom (GDOF) of 3-user $M \times N$ multiple-input multiple-output (MIMO) Gaussian interference channel (GIC) in which each transmitter has $M$ antennas and each receiver has $N$ antennas. A constraint of $2M \leq N$ is imposed so that random coding with message-splitting achieves the optimal GDOF. Unlike symmetric case, two cross channels to unintended receivers from each transmitter can have different strengths, and hence, well known Han-Kobayashi common-private message splitting would not achieve the optimal GDOF. Instead, splitting each user's message into three parts is shown to achieve the optimal GDOF. The capacity of the corresponding deterministic model is first established which provides systematic way of determining side information for converse. Although this deterministic model is philosophically similar to the one considered by Gou and Jafar, additional constraints are imposed so that capacity description of the deterministic model only contains the essential terms for establishing the GDOF of Gaussian case. Based on this, the optimal GDOF of Gaussian case is established with $\mathcal{O}(1)$ capacity approximation. The behavior of the GDOF is interestingly different from that of the corresponding symmetric case. Regarding the converse, several multiuser outer bounds which are suitable for asymmetric case are derived by non-trivial generalization of the symmetric case.
\end{abstract}
\section{Introduction}
\label{sec:intro}
Interference plays a central role in today's wireless communications systems. In information theory, the efforts of finding the performance limit of interference channel (IC) in terms of capacity started more than 30 years ago~\cite{Ca78,Sa81,HaKo81}. Unfortunately, the complete capacity region for even a simple 2-user IC is known only for \textit{strong} interference regime~\cite{Ca78,Sa81,HaKo81}. \\
\indent
Although the problem of finding the exact capacity has been open for more than 30 years, the notion of \textit{degrees of freedom} (DOF) defined for high signal-to-noise ratio (SNR) has opened a new direction of understanding IC. One surprising result was obtained by Cadambe and Jafar~\cite{CaJa08} which states that the per-user DOF of $K$-user IC is the same as that of 2-user IC for arbitrary $K$, which seems counter-intuitive given the fact that more users would result in more overall interference in the system. The DOF provides valuable understanding of IC with a form of conclusive answer, but it does not capture the relationship between signal strength and interference strength which has crucial importance in understanding IC. \\
\indent
In~\cite{EtTsWa08}, Etkin \textit{et al.} came up with the notion of the generalized DOF (GDOF) which incorporates signal-to-interference ratio (SIR) in it. As the DOF does, the GDOF also assumes high SNR, and this not only makes analysis more tractable, but also provides a valuable viewpoint in understanding IC. In IC, there are two important factors which are background noise and interference. Although their combined effect likely needs to be studied thoroughly for complete understanding of IC, one may want to isolate the effect of interference given that the effect of background noise has been fairly well studied through point-to-point (p2p) channel analysis. High SNR regime can essentially be considered as \textit{interference-limited} regime, and thus provides such isolation. It turns out that the GDOF provides tremendous insight on 2-user single-input single-output (SISO) Gaussian IC (GIC) through its so called `W' shape, and rather surprising 1-bit gap to the capacity result is also given in~\cite{EtTsWa08}.  \\
\indent
An important observation made in~\cite{EtTsWa08} is that a simple version of the Han-Kobayashi (HK) scheme~\cite{HaKo81} turns out to be the GDOF optimal. Intuition behind why the HK scheme is the GDOF optimal for 2-user SISO GIC can be found through its deterministic modeling which was originally studied by El Gamal and Costa~\cite{GaCo82}. Simply speaking, deterministic modeling assumes non-random noise or deterministic loss of transmitted signal which the transmitter is aware of. Therefore, the optimal strategy of the transmitter is easily given by not transmitting any valuable data on the part of the signal which will be lost. This strategy is indeed a special case of the HK scheme, and it is shown to achieve the capacity of this 2-user deterministic model. An important assumption in~\cite{GaCo82} for capacity achievability is that common information of interference must be clearly observable after decoding the intended message. This assumption is discussed in Section VII of~\cite{EtTsWa08}, and it will also be discussed later in this paper. By this assumption, a class of multiple-input multiple-output (MIMO) IC in which the HK scheme must be the GDOF optimal can be characterized. Gou and Jafar~\cite{GoJa11} found the optimal GDOF of a certain class of single-input multiple-output (SIMO) IC by extending the deterministic model of~\cite{GaCo82}. Corresponding MIMO results are obtained in~\cite{PaBlTa08, KaVa11, KaVa11-1, MoMu11}. \\
\indent 
For cases in which the HK scheme is not GDOF optimal, the optimal GDOF was found by using `signal-level alignment'~\cite{BrPaTs10, JaVi10}. For these cases, we may think of a specific form of deterministic modeling for Gaussian channels which are proposed in~\cite{BrTs08}. Although this `signal-level alignment' can possibly provide a valuable way of solving more general cases, it can only be applied for SISO symmetric cases so far. Extending this to general cases still remains to be seen. \\
\indent
One thing to note is that the aforementioned GDOF results only deal with symmetric IC except for 2-user results in \cite{EtTsWa08, KaVa11, KaVa11-1}. Since aforementioned assumption of clearly observable interference has nothing to do with symmetric nature of the channel, it is reasonable to believe that there must be a kind of message-splitting with random coding schemes which achieves the optimal GDOF of asymmetric IC, and this is the main focus of this paper. \\
\indent
A simpler case than asymmetric MIMO GIC was considered in~\cite{BrPaTs10}. In~\cite{BrPaTs10}, one-to-many IC was considered in which one transmitter causes interferences to all receivers, and all the other transmitters do not cause interference. In this channel, a generalization of the HK scheme which splits the message into multiple layers achieves a constant gap to the capacity. At the transmitter's view point, this one-to-many channel is equivalent to the channel considered in this paper. In this paper, therefore, we use this generalization of the HK scheme and show that it achieves the optimal GDOF of 3-user \textit{partially asymmetric} MIMO GIC. The reason why multiple splitting is necessary is because cross channels for a given transmitter have different channel qualities unlike symmetric case.\\
\indent
Finding the optimal GDOF involves derivation of tight-enough upper bounds. In~\cite{EtTsWa08}, a technique of giving appropriate side information is developed to derive such upper bounds. As mentioned in~\cite{GoJa11}, appropriate side information can easily be determined through deterministic modeling for certain cases. One thing to note is that the capacity region of the deterministic model given in~\cite{GoJa11} is much more complicated than the GDOF of the corresponding Gaussian model. For efficient computation, we propose more specific form of deterministic model which is closer to the corresponding Gaussian model for the GDOF analysis. By using this, the minimal number of tight upper bounds with appropriate side information can easily be determined. It will be seen that new type of upper bounds emerge, and they are non-trivial generalization of the symmetric case.\\
\indent
The remainder of this paper is organized as follows. Section~\ref{sec:ch} defines the channel model as well as achievable rate terms. Section~\ref{sec:det} provides analysis on the symmetric capacity of the deterministic models which correspond to the GDOF of Gaussian case. Section~\ref{sec:gau} establishes the optimal GDOF of the 3-user MIMO GIC. Section~\ref{sec:con} concludes the paper. 
\paragraph{Notation}
A matrix is represented with a capital letter like $X$, and a vector is represented as $\underline{x}$. $I$ represents an identity matrix or mutual information, and they can be easily differentiated from the context. For a matrix $X$ or a vector $\underline{x}$, $X^H$ or $\underline{x}^H$ represents conjugate transpose. $Tr(X)$ represents the trace of $X$. 
\section{Channel model and preliminaries}
\label{sec:ch}
Consider a following model with channel output $\underline{y}_i$ for the receiver $i$, channel input $\underline{x}_i$ for the transmitter $i$, and the channel $H_{ij}$ from the transmitter $i$ to the receiver $j$.
\begin{eqnarray}
\label{eq:ch}
\underline{y}_1&=&\rho H_{11}\underline{x}_1 + \rho^{\alpha_2} H_{21}\underline{x}_2 + \rho^{\alpha_1} H_{31}\underline{x}_3 + \underline{z}_1 \nonumber \\
\underline{y}_2&=&\rho^{\alpha_1} H_{12}\underline{x}_1 + \rho H_{22}\underline{x}_2 + \rho^{\alpha_2} H_{32}\underline{x}_3 + \underline{z}_2 \nonumber \\
\underline{y}_3&=&\rho^{\alpha_2} H_{13}\underline{x}_1 + \rho^{\alpha_1} H_{23}\underline{x}_2 + \rho H_{33}\underline{x}_3 + \underline{z}_3,
\end{eqnarray}
where background noise $\underline{z}_i \sim \mathcal{CN}(\underline{0},I)$ and $\rho>0, \alpha_1>\alpha_2>0$. $\underline{x}_i$ satisfies the average power constraint $Tr(E[\underline{x}_i\underline{x}_i^H]) \leq Tr(I)$. We consider $M \times N$ MIMO channel in which each transmitter has $M$ antennas and each receiver has $N$ antennas. Although every result obtained in this paper with 3 user can be directly generalized into $K$-user case, we only consider 3 user case in this paper due to computational complexity. We call the above model \textit{partially asymmetric} due to its symmetric nature that every transmitter sees channels with strengths $\rho, \rho^{\alpha_1}, \rho^{\alpha_2}$ and every receiver sees channels with strengths $\rho, \rho^{\alpha_1}, \rho^{\alpha_2}$. Again, a general asymmetric case is essentially no different from this partially asymmetric case, but we do not consider a general model due to complexity. Because of the symmetric nature of the channel, the achievable GDOF can be characterized by a single number as in the fully symmetric case while the asymmetric nature is enough to capture essential difference from the fully symmetric case. We assume that there is no degenerate case of channel coefficients, i.e., all $H_{ij}$'s are full-rank. We define the capacity region $\mathcal{C}$ of this channel in the standard Shannon sense. Because of symmetry, the maximum achievable total GDOF of the system is attained when the rate of each user is the same. Therefore, we define the symmetric capacity as 
\begin{equation}
C_{sym}= \max_{(R_1,R_2,R_3)\in \mathcal{C}} \min\{R_1, R_2, R_3\},
\end{equation}
where $R_i$ is the rate of user $i$. We may define $C_{sym}$ as a function of $\rho$, $\alpha_1$ and $\alpha_2$. Then, per user GDOF $d_{sym}(\alpha_1,\alpha_2)$ is given as 
\begin{equation}
d_{sym}(\alpha_1,\alpha_2)= \lim_{\rho \rightarrow \infty} \frac{C_{sym}(\rho,\alpha_1,\alpha_2)}{\log_2 \rho}.
\end{equation}
To satisfy the assumption of clearly observable interference, we only consider the case where $2M\leq N<3M$. This will be discussed in more detail in Section~\ref{sec:det}. 
\section{Deterministic modeling}
\label{sec:det}
\subsection{Case of $\alpha_1<1$}
Deterministic modeling gives an insight for a corresponding Gaussian model with simpler analysis. As mentioned earlier, one of the most important benefits of deterministic modeling is systematic determination of necessary side information. Figure~\ref{fig:det} shows the deterministic model corresponding to the channel defined in Section~\ref{sec:ch} with $\alpha_1<1$. There are two other possible cases of $\alpha_1>1>\alpha_2$ and $\alpha_2>1$, and slightly different deterministic models from one in Figure~\ref{fig:det} need to be considered for those cases. 
\begin{figure}[h]
\begin{center}{
 \includegraphics[width=0.5\textwidth]{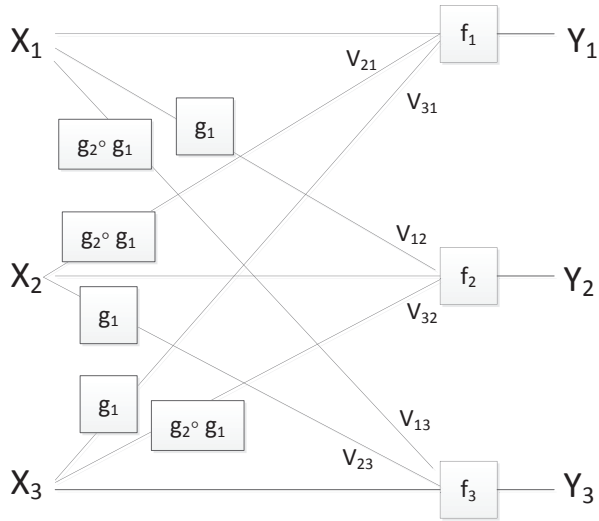}  
 }\end{center}
  \caption{Deterministic model of partially asymmetric IC}
  \label{fig:det}
\end{figure}
$V_{ij}$ is interference from the transmitter $i$ to the receiver $j$, $Y_i$ is channel output at the receiver $i$, and $X_i$ is channel input from the transmitter $i$, which are given as
\begin{subequations} 
\label{eq:det1}
\begin{eqnarray}
V_{12}&=&g_1(X_1), \quad V_{13}=g_2\circ g_1(X_1),\\
V_{21}&=&g_2\circ g_1(X_2), \quad V_{23}=g_1(X_2),\\
V_{31}&=&g_1(X_3), \quad V_{32}=g_2\circ g_1(X_3),\\
Y_1&=&f_1(X_1,V_{21},V_{31}),\\
Y_2&=&f_2(X_2,V_{12},V_{32}),\\
Y_3&=&f_3(X_3,V_{13},V_{23}),
\end{eqnarray}
\end{subequations}
where $f_i$ and $g_i$ are deterministic functions. The term `deterministic' comes from this property of channel functions especially $g_i$. This is similar to what is defined in~\cite{GoJa11}, but an important difference is that the functions representing two interference channels from one transmitter are different, which reflects asymmetric nature of the channel. In~\cite{GoJa11}, functions representing interference channels from one transmitter are different from those from the other transmitters, and in that sense, the model in~\cite{GoJa11} is more general than the one described in Figure~\ref{fig:det}. Note that a model with the same functions from all the transmitters would have been enough to show the intended results of Gaussian case in~\cite{GoJa11}, and we only consider this simpler model which is enough due to symmetry of the channel. Another difference from~\cite{GoJa11} is that one interference from each transmitter is a degraded version of the other interference. Because of degraded nature of Gaussian interference channel which is explained in~\cite[Ch. 15.6.3]{CoTh06} for broadcast channel, this deterministic model is sufficient to reflect the GDOF behavior of Gaussian model. If we consider a more general deterministic model, then superposition coding which is enabled by degraded nature of the channel would likely be insufficient to achieve the capacity, and the resulting capacity achieving scheme, if possible to find, would look quite different from the GDOF achieving scheme for Gaussian IC. For that reason, we only consider a degraded deterministic model. An important property which needs to be satisfied to show that the HK-like scheme achieves the capacity is given as
\begin{subequations} 
\label{eq:intdec}
\begin{eqnarray}
H(Y_1|X_1)&=&H(V_{21},V_{31})=H(V_{21})+H(V_{31}),\\
H(Y_2|X_2)&=&H(V_{12},V_{32})=H(V_{12})+H(V_{32}),\\
H(Y_3|X_3)&=&H(V_{13},V_{23})=H(V_{13})+H(V_{23}).
\end{eqnarray}
\end{subequations}
Note that the second equality of each line in the above equation automatically holds due to independence, and hence the assumption essentially is the first inequality of each line. When this holds, each interference is decodable given the intended message which implies that there is enough dimension to resolve interference uncertainty. \eqref{eq:intdec} is equivalent to restricting function $f_i$ from $(v_{ji},v_{ki})$ to $y_i$ for given $x_i$ to be injective. In MIMO Gaussian case, it is not difficult to see that $N\geq 2M$ must be satisfied to ensure enough dimension although there is no formal proof that the HK-like scheme will not be GDOF optimal when this does not hold. We also consider only the case of $N<3M$ since the DOF of $M$ per user can be easily achieved by IAN if $N \geq 3M$.\\
\indent
As mentioned earlier, the capacity region of the deterministic model in~\cite{GoJa11} is considerably more complicated than the GDOF of the corresponding Gaussian model. This is due to the deterministic model being more general than the corresponding Gaussian model. Although the capacity result of a general deterministic model has value by itself, we will consider a special case of the deterministic model which resembles Gaussian model more closely to reduce amount of analysis. The deterministic model described in Figure~\ref{fig:det} additionally satisfies the following properties. 
For all $i$ and $j$, let
\begin{eqnarray}
\label{eq:set}
A_{ij}&=&\begin{cases}
\{ V_{ji}\} & \text{ if } V_{ji}=g_2 \circ g_1 (X_j)\\
\{ V_{jk} \text{ for all } k \} & \text{ if } V_{ji}= g_1 (X_j)\\
\{  X_j, V_{jk} \text{ for all } k \} & \text{ if } V_{ji} \text{ does not exist.}
\end{cases}
\end{eqnarray}
Simply speaking, $A_{ij}$ is the set of messages from transmitter $j$ which need to be decoded at receiver $i$.
Let $A$ be a subset of the set of messages $\{V_{12}, V_{13}, V_{21}, V_{23}, V_{31}, V_{32}\}$.
Then, we have for all $i$ and $j$ such that $V_{ij}=g_1 (X_i)$
\begin{eqnarray}
\label{eq:self}
I(V_{ij};Y_i|A)&=&I(V_{ij};Y_i| V_{ki} \text{ for all }k, A \cap A_{ii})=H(V_{ij}|A \cap A_{ii}),
\end{eqnarray}
if $A_{il} \subset A$ for $l \neq i$. This condition means that the maximum transmittable rate of the message $V_{ij}$ from transmitter $i$ to receiver $i$ is not changed by giving one of interferer's message to receiver $i$ as side information if another interferer's message is already given to receiver $i$. Since we only consider $N\geq 2M$ in Gaussian case, presence of interference whose dimension is $M$ does not affect decodability of $M$ dimensional message much in high SNR regime. This phenomenon is essentially described in Lemma~\ref{lem:dim}, and it governs the GDOF behavior. \eqref{eq:self} assumes similar condition in deterministic model such that the symmetric capacity of this model more closely resembles the GDOF of Gaussian case. We also assumes the following. For all $i$ and $j$ and $k \neq i,j$, we have
\begin{eqnarray}
\label{eq:int}
I(V_{ji};Y_i|A)&=&I(V_{ji};Y_i|X_i,V_{ki}, A\cap A_{ij})=H(V_{ji}| A\cap  A_{ij}),
\end{eqnarray}
if $A_{ik} \subset A$. This condition is the counter part of~\eqref{eq:self} for the message from transmitter $j$ to receiver $i$. It will be seen that these assumptions result in significant reduction of analysis and give right amount of insight for Gaussian case. 
As will be seen in Section~\ref{sec:gau}, the GDOF of Gaussian case is irrelevant to the input covariance as long as input is Gaussian satisfying power constraint. The capacity region of the deterministic model, however, would depend on input distribution. Again, to reduce amount of analysis, we only consider the case when $p(x_1)=p(x_2)=p(x_3)=p(x)$. 
\begin{theorem}
\label{thm:detcap}
The symmetric capacity of the deterministic model given in~\eqref{eq:det1} with $p(q,x_1,x_2,x_3)=p(q)p(x_1|q)$\\
$\times p(x_2|q)p(x_3|q)$ and $p(x_1|q)=p(x_2|q)=p(x_3|q)=p(x|q)$ is given as
\begin{eqnarray}
\label{eq:detcap}
&&C_{sym}\nonumber\\
&&=\min\Bigg\{ I(V_{31},V_{21};Y_1|V_{12},V_{32},Q)+I(X_1;Y_1|V_{12},V_{21},V_{31},Q),\nonumber \\
           && \qquad \qquad\frac{I(V_{21};Y_1|V_{12},V_{31},Q)+I(V_{12},V_{21},V_{31};Y_1|V_{13},V_{32},Q)}{2}+I(X_1;Y_1|V_{12},V_{21},V_{31},Q),\nonumber \\
           && \qquad \qquad\frac{I(V_{21},V_{31};Y_1|V_{12},Q)+I(V_{12};Y_1|V_{13},V_{21},V_{31},Q)}{2}+I(X_1;Y_1|V_{12},V_{21},V_{31},Q),\nonumber \\
           && \qquad \qquad\frac{1}{2}I(V_{12},V_{21},V_{31};Y_1|V_{32},Q)+I(X_1;Y_1|V_{12},V_{21},V_{31},Q),\nonumber \\
           && \qquad \qquad\frac{1}{2} I(V_{12},V_{21},V_{31};Y_1|V_{13},Q) +I(X_1;Y_1|V_{12},V_{21},V_{31},Q),\nonumber \\
           && \qquad \qquad\frac{I(V_{12},V_{21},V_{31};Y_1|Q)+I(V_{12};Y_1|V_{13},V_{21},V_{31},Q)}{3}+I(X_1;Y_1|V_{12},V_{21},V_{31},Q) \Bigg\}.
\end{eqnarray}
\end{theorem}
\begin{proof}
\begin{enumerate}
\item Achievability
\begin{enumerate}
\item Codebook generation\\
Given $p(x_i|q)$, the joint probability mass function of $p(x_i,v_{ij},v_{ik}|q)$ can be defined where $v_{ij}=g_1(x_i)$ and $v_{ik}=g_2 \circ g_1(x_i)$ . First, a time-sharing sequence $q^n$ is generated by choosing each element independently according to $p(q)$. This $q^n$ is shared by all transmitters and receivers. Transmitter $i$ generates $2^{nR_{c_1}}$ codewords $v^n_{ik}(l_{ik}), l_{ik} \in \{1,2,...,2^{nR_{c_1}}\}$ of length $n$ by selecting the $m$th element of each codeword according to $p(v_{ik}|q^n_m)$ where $q^n_m$ is the $m$th element of $q^n$. For each codeword $v^n_{ik}(l_{ik})$, transmitter $i$ generates $2^{nR_{c_2}}$ codewords $v^n_{ij}(l_{ik}, l_{ij}), l_{ij} \in \{1,2,...,2^{nR_{c_2}}\}$ of length $n$ by selecting the $m$th element of each codeword according to $p(v_{ij}|v^n_{ik,m}(l_{ik}), q^n_m)$ where $v^n_{ik,m}(l_{ik})$ is the $m$th element of $v^n_{ik}(l_{ik})$. For each codeword $v^n_{ij}(l_{ik}, l_{ij})$, transmitter $i$ generates $2^{nR_{p}}$ codewords $x^n_{i}(l_{ik}, l_{ij},l_i), l_{i} \in \{1,2,...,2^{nR_{p}}\}$ of length $n$ by selecting the $m$th element of each codeword according to $p(x_i|v^n_{ik,m}(l_{ik}), v^n_{ij,m}(l_{ij}),q^n_m)$. Note that this codebook generation implies that $R_1=R_2=R_3=R=R_{c_1}+R_{c_2}+R_p$. We only consider this symmetric rate allocation since we are interested in the maximum sum rate. Note that the maximum sum rate is obtained by the symmetric rate allocation due to symmetric nature of the channel. 
\item Encoding\\
Transmitter $i$ sends codeword $x^n_{i}(l_{ik}, l_{ij},l_i)$ corresponding to the message indexed by $(l_{ik}, l_{ij},l_i)$.
\item Decoding\\
Decoding is done by checking typicality. Detailed mathematical description about decoding will not be considered here since it would require proper definition of every quantity involved which could be exhausting. Detailed description about typical decoding can be found in~\cite{CoTh06}. We also assume successive decoding. All common messages are jointly decoded first while treating all private messages as noise, and the private message of the intended transmitter is decoded successively by considering other private messages as noise. Note that $V_{jk}, k \neq i$ for $j$ such that $V_{ji}=g_2 \circ g_1(X_j)$ is not common message at receiver $i$, i.e., $V_{23}$ is not common message at receiver 1. 
\item Error analysis\\
By using standard error event analysis for typical decoding we get the following rate bounds at receiver 1. Note that we only need to consider receiver 1 due to symmetry.
\begin{subequations} 
\begin{eqnarray}
R_p&<&I(X_1;Y_1|V_{12},V_{21},V_{31},Q)\\
R_{c2}&<& I(V_{12};Y_1|V_{13},V_{21},V_{31},Q)\\
R_{c1}&<& I(V_{21};Y_1|V_{12},V_{31},Q)\\
R_{c2}&<& I(V_{31};Y_1|V_{12},V_{21},V_{32},Q)\\
R_{c1}+R_{c2}&<& I(V_{12};Y_1|V_{21},V_{31},Q)\\
R_{c1}+R_{c2}&<& I(V_{31};Y_1|V_{12},V_{21},Q)\\
R_{c1}+R_{c2}&<& I(V_{12},V_{21};Y_1|V_{13},V_{31},Q)\\
R_{c1}+R_{c2}&<& I(V_{31},V_{21};Y_1|V_{12},V_{32},Q)\\
R_{c1}+R_{c2}&<& I(V_{12},V_{31};Y_1|V_{13},V_{21},V_{32},Q)\\
2R_{c1}+R_{c2}&<& I(V_{12},V_{21};Y_1|V_{31},Q)\\
R_{c1}+2R_{c2}&<& I(V_{12},V_{31};Y_1|V_{21},V_{32},Q)\\
R_{c1}+2R_{c2}&<& I(V_{12},V_{21},V_{31};Y_1|V_{13},V_{32},Q)\\
R_{c1}+2R_{c2}&<& I(V_{12},V_{31};Y_1|V_{13},V_{21},Q)\\
2R_{c1}+R_{c2}&<& I(V_{21},V_{31};Y_1|V_{12},Q)\\
2R_{c1}+2R_{c2}&<& I(V_{12},V_{31};Y_1|V_{21},Q)\\
2R_{c1}+2R_{c2}&<& I(V_{12},V_{21},V_{31};Y_1|V_{32},Q)\\
2R_{c1}+2R_{c2}&<& I(V_{12},V_{21},V_{31};Y_1|V_{13},Q)\\
3R_{c1}+2R_{c2}&<& I(V_{12},V_{21},V_{31};Y_1|Q).
\end{eqnarray}
\end{subequations}
By using~\eqref{eq:self} and~\eqref{eq:int}, we can further reduce the number of relevant bounds for the symmetric rate as 
\begin{subequations} 
\label{eq:subnd}
\begin{eqnarray}
\label{eq:subnd1}
R_p&<&I(X_1;Y_1|V_{12},V_{21},V_{31},Q)\\
\label{eq:subnd2}
R_{c2}&<& I(V_{12};Y_1|V_{13},V_{21},V_{31},Q)\\
\label{eq:subnd3}
R_{c1}&<& I(V_{21};Y_1|V_{12},V_{31},Q)\\
\label{eq:subnd4}
R_{c1}+R_{c2}&<& I(V_{31},V_{21};Y_1|V_{12},V_{32},Q)\\
\label{eq:subnd5}
R_{c1}+2R_{c2}&<& I(V_{12},V_{21},V_{31};Y_1|V_{13},V_{32},Q)\\
\label{eq:subnd6}
2R_{c1}+R_{c2}&<& I(V_{21},V_{31};Y_1|V_{12},Q)\\
\label{eq:subnd7}
2R_{c1}+2R_{c2}&<& I(V_{12},V_{21},V_{31};Y_1|V_{32},Q)\\
\label{eq:subnd8}
2R_{c1}+2R_{c2}&<& I(V_{12},V_{21},V_{31};Y_1|V_{13},Q)\\
\label{eq:subnd9}
3R_{c1}+2R_{c2}&<& I(V_{12},V_{21},V_{31};Y_1|Q).
\end{eqnarray}
\end{subequations}
The above inequalities can be written as 
\begin{subequations} 
\label{eq:sumbnd}
\begin{eqnarray}
\eqref{eq:subnd4}+\eqref{eq:subnd1}: R&<&I(V_{31},V_{21};Y_1|V_{12},V_{32},Q)+I(X_1;Y_1|V_{12},V_{21},V_{31},Q)\\
\eqref{eq:subnd3}+\eqref{eq:subnd5}+2 \times \eqref{eq:subnd1}:2R&<&I(V_{21};Y_1|V_{12},V_{31},Q)+I(V_{12},V_{21},V_{31};Y_1|V_{13},V_{32},Q)\nonumber\\
&&\quad +2I(X_1;Y_1|V_{12},V_{21},V_{31},Q)\\      
\eqref{eq:subnd6}+\eqref{eq:subnd2}+2 \times \eqref{eq:subnd1}:2R&<&I(V_{21},V_{31};Y_1|V_{12},Q)+I(V_{12};Y_1|V_{13},V_{21},V_{31},Q)\nonumber\\
&&\quad +2I(X_1;Y_1|V_{12},V_{21},V_{31},Q)\\
\eqref{eq:subnd7}+2 \times \eqref{eq:subnd1}:2R&<& I(V_{12},V_{21},V_{31};Y_1|V_{32},Q)+2I(X_1;Y_1|V_{12},V_{21},V_{31},Q)\\
\eqref{eq:subnd8}+2 \times \eqref{eq:subnd1}:2R&<& I(V_{12},V_{21},V_{31};Y_1|V_{13},Q)+2I(X_1;Y_1|V_{12},V_{21},V_{31},Q)\\
\eqref{eq:subnd9}+\eqref{eq:subnd2}+3 \times \eqref{eq:subnd1}:3R&<&I(V_{12},V_{21},V_{31};Y_1|Q)+I(V_{12};Y_1|V_{13},V_{21},V_{31},Q)\nonumber\\
&&\quad +3I(X_1;Y_1|V_{12},V_{21},V_{31},Q),
\end{eqnarray}
\end{subequations} 
which determines the expression for the maximum symmetric rate given as~\eqref{eq:detcap}.
\end{enumerate}
\item Converse\\
It is sufficient to show that $R_1+R_2+R_3$ must be smaller than $3C_{sym}$ for reliable communication which corresponds to show that $R_1+R_2+R_3$ must be smaller than 3 times of each term in~\eqref{eq:detcap}. Let us consider the last term of ~\eqref{eq:detcap} and rewrite it as
\begin{subequations}
\begin{eqnarray}
&&R_1+R_2+R_3\nonumber \\
&&\quad < I(V_{12},V_{21},V_{31};Y_1|Q)+I(V_{12};Y_1|V_{13},V_{21},V_{31},Q)+3I(X_1;Y_1|V_{12},V_{21},V_{31},Q)\\
&&\quad=I(V_{12},V_{21},V_{31};Y_1|Q)+I(V_{23};Y_2|V_{12},V_{21},V_{32},Q)\nonumber\\
&&\qquad+I(X_1;Y_1|V_{12},V_{21},V_{31},Q)+I(X_2;Y_2|V_{12},V_{23},V_{32},Q)+I(X_3;Y_3|V_{13},V_{23},V_{31},Q)\\
&& \quad = H(Y_1|Q)+H(Y_2|V_{12},V_{21},V_{32},Q)+H(Y_3|V_{13},V_{23},V_{31},Q).
\end{eqnarray}
\end{subequations}
The above expression implies that receiver 2 needs to have side information $V_{12},V_{21},V_{32}$, and receiver 3 needs side information $V_{13},V_{23},V_{31}$. Now we proceed to converse. From Fano's inequality, we have
\begin{subequations}
\begin{eqnarray}
&&n(R_1+R_2+R_3)\nonumber \\
&&\quad\leq I(X_1^n;Y_1^n)+I(X_2^n;Y_2^n)+I(X_3^n;Y_3^n)\\
           &&\quad \leq I(X_1^n;Y_1^n)+I(X_2^n;Y_2^n,V^n_{12},V^n_{21},V^n_{32})+I(X_3^n;Y_3^n,V^n_{13},V^n_{23},V^n_{31})\\
           &&\quad = H(Y_1^n)-H(V^n_{21})-H(V^n_{31})+H(V^n_{21})+H(Y_2^n|V^n_{12},V^n_{21},V^n_{32})\nonumber\\
           && \qquad +H(V^n_{31})+H(Y_3^n|V^n_{13},V^n_{23},V^n_{31})\\
           &&\quad =H(Y_1^n)+H(Y_2^n|V^n_{12},V^n_{21},V^n_{32})+H(Y_3^n|V^n_{13},V^n_{23},V^n_{31})\\
           &&\quad \leq \sum_{i=1}^n \Big (H(Y_{1i})+H(Y_{2i}|V_{12i},V_{21i},V_{32i})+H(Y_{3i}|V_{13i},V_{23i},V_{31i})\Big)\\
           &&= \sum_{i=1}^n \Big (H(Y_{1q}|Q=i)+H(Y_{2q}|V_{12q},V_{21q},V_{32q},Q=i)+H(Y_{3q}|V_{13q},V_{23q},V_{31q},Q=i)\Big)\\
           &&\quad = n\Big(H(Y_1|Q)+H(Y_2|V_{12},V_{21},V_{32},Q)+H(Y_3|V_{13},V_{23},V_{31},Q)\Big).
\end{eqnarray}
\end{subequations}
,where $Q=i \in \{1,2,...,n\}$ with probability $1/n$, and $Y_1=Y_{1Q}$ and all other similar terms are new random variables whose distributions depend on $Q$ in the same way as the distributions of $Y_{1i}$ and all other similar terms depend on $i$. We can do similarly for each element of~\eqref{eq:detcap} which completes converse.
\end{enumerate}
\end{proof}
In the proof of the above theorem, successive decoding is used. The optimal way of decoding would be jointly decoding every message which needs to be decoded, i.e., $X_1, V_{12}, V_{13}, V_{21}, V_{31}, V_{32}$ needs to be jointly decoded at receiver 1. It is not difficult to see that every bound of~\eqref{eq:sumbnd} must be a bound for joint decoding as well, and hence, joint decoding is not better than successive decoding in terms of the symmetric capacity. In fact, joint decoding is not better even in terms of achievable region, i.e., successive decoding achieves the capacity region. The key reason is from sequential superposition encoding at the transmitter. Suppose $x_1^n(1,1,1)$ was transmitted. An error event of decoded message $x_1^n(i,j,1)$ for $i,j \neq 1$ cannot be evaluated differently from an error event of decoded message $x_1^n(i,j,k)$ for $i,j,k \neq 1$, since no part of $x_1^n(i,j,1)$ is the actually transmitted message. As a result, error events which can be evaluated with joint decoding are the same as those with successive decoding. In this deterministic model, however, all common messages still need to be jointly decoded, since only messages from the same transmitter are superposition encoded. Successive decoding of common and private messages are allowed even in this case due to the fact that incorrect decoding of interferers' messages with correct decoding of intended messages is not an error in interference channel. \\
\indent
The reason why sequential superposition encoding is needed in the proof of the above theorem is due to degraded nature of the deterministic model. It is well known that sequential superposition encoding is optimal for degraded broadcast channel (BC)~\cite[Ch. 15.6.3]{CoTh06}, and this carries over to the deterministic model in this paper. One thing to note is that simultaneous superposition encoding can be used in Gaussian BC even though every Gaussian BC is degraded~\cite[Ch. 15.6.3]{CoTh06}. Since simultaneous superposition encoding is used for MIMO GIC in this paper, successive decoding does not fully achieve the capacity region, which implies that terms in~\eqref{eq:detcap} would not fully cover bounds defining the GDOF region. In this paper, however, we are interested in the GDOF defined by the symmetric capacity, and hence this does not cause a problem. In the same reason, successive decoding will also be used for GIC. \\
\indent
As seen in the proof of the above theorem, the form of achievability bounds in~\eqref{eq:sumbnd} directly determines side information for each receiver. This is one of the benefits of deterministic modeling as mentioned earlier. As will be seen, we can directly apply this to corresponding Gaussian model. We should also note that converse is enabled by the assumption given in~\eqref{eq:intdec}.
\subsection{Case of $\alpha_2<1<\alpha_1$}
The corresponding deterministic model is given as follows.
\begin{subequations} 
\label{eq:det2}
\begin{eqnarray}
V_{11}&=&g_1(X_1), \quad V_{13}=g_2\circ g_1(X_1),\\
V_{21}&=&g_2\circ g_1(X_2), \quad V_{22}=g_1(X_2),\\
V_{32}&=&g_2\circ g_1(X_3), \quad V_{33}=g_1(X_3),\\
Y_1&=&f_1(V_{11},V_{21},X_3),\\ 
Y_2&=&f_2(V_{22},X_1,V_{32}),\\
Y_3&=&f_3(V_{33},V_{13},X_2),
\end{eqnarray}
\end{subequations}
where $f_i$ and $g_i$ are deterministic functions. We assume 
\begin{subequations} 
\label{eq:intdec2}
\begin{eqnarray}
H(Y_1|V_{11})&=&H(V_{21},X_3)=H(V_{21})+H(X_3),\\
H(Y_2|V_{22})&=&H(X_1,V_{32})=H(X_1)+H(V_{32}),\\
H(Y_3|V_{33})&=&H(V_{13},X_2)=H(V_{13})+H(X_2).
\end{eqnarray}
\end{subequations}
As in the case of $\alpha<1$, we assume the following. First, define $A_{ij}$ as in~\eqref{eq:set}.
Let $A$ be a subset of sets of messages $\{X_1, V_{11}, V_{13}, X_2, V_{22}, V_{21}, X_3, V_{32}, V_{33}\}$.
Then, we have for all $i$ and $j$ such that $X_j \in A_{ij}$
\begin{eqnarray}
\label{eq:strong}
I(X_j;Y_i|A)&=&I(X_j;Y_i|V_{ki} \text{ for all }k, A \cap A_{ij})=H(X_j|A \cap A_{ij}),
\end{eqnarray}
if $A_{il} \subset A$ for $l \neq j$. For all $i$ and $j$ and $k,l$ such that $X_k\in A_{ik}$ and $l\neq k,j$, we have
\begin{eqnarray}
\label{eq:weak}
I(V_{ji};Y_i|A)&=&I(V_{ji};Y_i|V_{li}, X_{k}, A\cap A_{ij})=H(V_{ji}| A\cap  A_{ij}),
\end{eqnarray}
if $A_{ik} \subset A$ or $A_{il} \subset A$. We now present the symmetric capacity of the deterministic model.
\begin{theorem}
\label{thm:detcap2}
The symmetric capacity of the deterministic model given in~\eqref{eq:det2} with $p(q,x_1,x_2,x_3)=p(q)p(x_1|q)$\\
$\times p(x_2|q)p(x_3|q)$ and $p(x_1|q)=p(x_2|q)=p(x_3|q)=p(x|q)$ is given as
\begin{eqnarray}
\label{eq:detcap2}
&&C_{sym}\nonumber\\
&&=\min\Bigg\{ I(V_{11},V_{21};Y_1|V_{13},X_3,Q),\nonumber\\
&& \qquad \qquad \frac{I(V_{21};Y_1|V_{11},X_3,Q)+I(V_{11},V_{21},X_3;Y_1|V_{13},V_{32},Q)}{2},\nonumber \\
&& \qquad \qquad \frac{1}{2}I(V_{11},V_{21},X_3;Y_1|V_{32},Q),\nonumber \\
&& \qquad \qquad \frac{1}{2}I(V_{11},V_{21},X_3;Y_1|V_{13},Q),\nonumber\\
&& \qquad \qquad\frac{I(V_{11},V_{21},X_3;Y_1|Q)+I(V_{11};Y_1|V_{13},V_{21},X_3,Q)}{3}\Bigg\}.
\end{eqnarray}
\end{theorem}
\begin{proof}
\begin{enumerate}
\item Achievability
\begin{enumerate}
\item Codebook generation\\
Given $p(x_i|q)$, the joint probability mass function of $p(x_i,v_{ii},v_{ij}|q)$ can be defined where $v_{ii}=g_1(x_i)$ and $v_{ij}=g_2 \circ g_1(x_i)$ . First, a time-sharing sequence $q^n$ is generated by choosing each element independently according to $p(q)$. This $q^n$ is shared by all transmitters and receivers. Transmitter $i$ generates $2^{nR_{c_1}}$ codewords $v^n_{ij}(l_{ij}), l_{ij} \in \{1,2,...,2^{nR_{c_1}}\}$ of length $n$ by selecting the $m$th element of each codeword according to $p(v_{ij}|q^n_m)$. For each codeword $v^n_{ij}(l_{ij})$, transmitter $i$ generates $2^{nR_{c_2}}$ codewords $v^n_{ii}(l_{ij}, l_{ii}), l_{ii} \in \{1,2,...,2^{nR_{c_2}}\}$ of length $n$ by selecting the $m$th element of each codeword according to $p(v_{ii}|v^n_{ij,m}(l_{ij}),q^n_m)$. For each codeword $v^n_{ii}(l_{ij}, l_{ii})$, transmitter $i$ generates one codeword $x^n_{i}(l_{ij}, l_{ii})$ by selecting the $m$th element of each codeword according to $p(x_i|v^n_{ij,m}(l_{ij}), v^n_{ii,m}(l_{ii}),q^n_m)$.
\item Encoding\\
Transmitter $i$ sends codeword $x^n_{i}(l_{ij}, l_{ii})$ corresponding to the message indexed by $(l_{ij}, l_{ii})$.
\item Decoding\\
All messages are decoded jointly, i.e., $V_{11},V_{13},V_{21},X_{3},V_{33},V_{32}$ are jointly decoded at the receiver 1. 
\item Error analysis\\
By using standard error event analysis for typical decoding and~\eqref{eq:strong} and~\eqref{eq:weak}, we get the relevant bounds for the symmetric rate as 
\begin{subequations} 
\begin{eqnarray}
R_{c2}&<& I(V_{11};Y_1|V_{13},V_{21},X_3,Q)\\
R_{c1}&<& I(V_{21};Y_1|V_{11},X_3,Q)\\
R_{c1}+R_{c2}&<& I(V_{11},V_{21};Y_1|V_{13},X_3,Q)\\
R_{c1}+2R_{c2}&<& I(V_{11},V_{21},X_3;Y_1|V_{13},V_{32},Q)\\
2R_{c1}+2R_{c2}&<& I(V_{11},V_{21},X_3;Y_1|V_{32},Q)\\
2R_{c1}+2R_{c2}&<& I(V_{11},V_{21},X_3;Y_1|V_{13},Q)\\
3R_{c1}+2R_{c2}&<& I(V_{11},V_{21},X_3;Y_1|Q).
\end{eqnarray}
\end{subequations}
The above inequalities can be written as 
\begin{subequations} 
\begin{eqnarray}
R&<&I(V_{11},V_{21};Y_1|V_{13},X_3,Q)\\
2R&<& I(V_{21};Y_1|V_{11},X_3,Q)+I(V_{11},V_{21},X_3;Y_1|V_{13},V_{32},Q)\\
2R&<&I(V_{11},V_{21},X_3;Y_1|V_{32},Q)\\
2R&<&I(V_{11},V_{21},X_3;Y_1|V_{13},Q)\\
3R&<&I(V_{11},V_{21},X_3;Y_1|Q)+I(V_{11};Y_1|V_{13},V_{21},X_3,Q),
\end{eqnarray}
\end{subequations} 
which determines the expression for the maximum symmetric rate given as~\eqref{eq:detcap2}.
\end{enumerate}
\item Converse\\
Converse is proven in similar ways to the case of $\alpha_1<1$, i.e., show that $R_1+R_2+R_3$ must be smaller than 3 times of each term in~\eqref{eq:detcap2}. There are several bounds which are proven through slightly different ways from the case of $\alpha_1<1$. Let us consider the third term of~\eqref{eq:detcap2}. It can be written as
\begin{subequations}
\begin{eqnarray}
&&2(R_1+R_2+R_3)\nonumber \\
&&\quad <I(V_{11},V_{21},X_3;Y_1|V_{32},Q)+I(X_1,V_{22},V_{32};Y_2|V_{13},Q)+I(V_{13},X_2,V_{33};Y_3|V_{21},Q) \\
&& \quad = H(Y_1|V_{32},Q)+H(Y_2|V_{13},Q)+H(Y_1|V_{21},Q).
\end{eqnarray}
\end{subequations}
Although the above expression only implies that receiver 1 needs to have side information $V_{32}$, the actual converse requires more effort than that. From Fano's inequality, we have
\begin{subequations}
\begin{eqnarray}
&&2n(R_1+R_2+R_3)\nonumber \\
&&\quad\leq 2I(X_1^n;Y_1^n)+2I(X_2^n;Y_2^n)+2I(X_3^n;Y_3^n)\\
 &&\quad \leq I(X_1^n;Y_1^n,V^n_{32})+I(X_1^n;Y_1^n,X_1^n)+I(X_2^n;Y_2^n,V^n_{13})+I(X_2^n;Y_2^n,X_2^n)\nonumber\\
 &&\qquad +I(X_3^n;Y_3^n,V^n_{21})+I(X_3^n;Y_3^n,X_3^n)\\
 &&\quad =H(Y_1^n|V^n_{32})-H(V^n_{21})-H(X_3^n|V^n_{32})+H(X_1^n)\nonumber\\
 && \qquad +H(Y_2^n|V^n_{13})-H(V^n_{32})-H(X_1^n|V^n_{13})+H(X_2^n)\nonumber\\
  && \qquad+H(Y_3^n|V^n_{21})-H(V^n_{13})-H(X_2^n|V^n_{21})+H(X_3^n)\nonumber\\
  &&\quad =H(Y_1^n|V^n_{32})+H(Y_2^n|V^n_{13})+H(Y_3^n|V^n_{21})\\
  &&\quad \leq \sum_{i=1}^n \Big (H(Y_{1i}|V_{32i})+H(Y_{2i}|V_{13i})+H(Y_{3i}|V_{21i})\Big)\\
           &&\quad = n\Big(H(Y_1|V_{32},Q)+H(Y_2|V_{13},Q)+H(Y_3|V_{21},Q)\Big).
\end{eqnarray}
\end{subequations}
Remaining bounds can be evaluated similarly.
\end{enumerate}
\end{proof}
\subsection{Case of $\alpha_2>1$}
The corresponding deterministic model is given as follows.
\begin{subequations} 
\label{eq:det3}
\begin{eqnarray}
V_{11}&=&g_2\circ g_1(X_1), \quad V_{13}=g_1(X_1),\\
V_{21}&=&g_1(X_2), \quad V_{22}=g_2\circ g_1(X_2),\\
V_{32}&=&g_1(X_3), \quad V_{33}=g_2\circ g_1(X_3),\\
Y_1&=&f_1(V_{11},V_{21},X_3),\\ 
Y_2&=&f_2(V_{22},X_1,V_{32}),\\
Y_3&=&f_3(V_{33},V_{13},X_2),
\end{eqnarray}
\end{subequations}
where $f_i$ and $g_i$ are deterministic functions. 
We assume~\eqref{eq:intdec2},~\eqref{eq:strong},~\eqref{eq:weak} in exactly the same ways to the case of $\alpha_2<1<\alpha_1$.
We now present the symmetric capacity of the deterministic model.
\begin{theorem}
\label{thm:detcap3}
The symmetric capacity of the deterministic model given in~\eqref{eq:det3} with $p(q,x_1,x_2,x_3)=p(q)p(x_1|q)$\\
$\times p(x_2|q)p(x_3|q)$ and $p(x_1|q)=p(x_2|q)=p(x_3|q)=p(x|q)$ is given as
\begin{eqnarray}
\label{eq:detcap3}
&&C_{sym}\nonumber\\
&&=\min\Bigg\{ I(V_{11};Y_1|V_{21},X_3,Q), \frac{1}{3} I(V_{11},V_{21},X_3;Y_1|Q)\Bigg\}.
\end{eqnarray}
\end{theorem}
\begin{proof}
\begin{enumerate}
\item Achievability
\begin{enumerate}
\item Codebook generation\\
Given $p(x_i|q)$, the joint probability mass function of $p(x_i,v_{ii},v_{ij})$ can be defined where $v_{ii}=g_2 \circ g_1(x_i)$ and $v_{ij}=g_1(x_i)$. First, a time-sharing sequence $q^n$ is generated by choosing each element independently according to $p(q)$. This $q^n$ is shared by all transmitters and receivers. Transmitter $i$ generates $2^{nR}$ codewords $v^n_{ii}(l_{ii}), l_{ii} \in \{1,2,...,2^{nR}\}$ of length $n$ by selecting the $m$th element of each codeword independently according to $p(v_{ii}|q^n_m)$. For each codeword $v^n_{ii}(l_{ii})$, transmitter $i$ generates one codeword $v^n_{ij}(l_{ii})$ of length $n$ by selecting the $m$th element of each codeword according to $p(v_{ij}|v^n_{ii,m}(l_{ii}),q^n_m)$. For each codeword $v^n_{ij}(l_{ii})$, transmitter $i$ generates one codeword $x^n_{i}(l_{ii})$ by selecting the $m$th element of each codeword according to $p(x_i|v^n_{ii,m}(l_{ii}), v^n_{ij,m}(l_{ii}),q^n_m)$.
\item Encoding\\
Transmitter $i$ sends codeword $x^n_{i}(l_{ii})$ corresponding to the message indexed by $l_{ii}$.
\item Decoding\\
All messages are decoded jointly, i.e., $V_{11},V_{21},V_{22},X_{3},V_{33},V_{32}$ are jointly decoded at the receiver 1. 
\item Error analysis\\
By using standard error event analysis for typical decoding and~\eqref{eq:strong},~\eqref{eq:weak}, we get the relevant bounds for the symmetric rate as 
\begin{subequations} 
\begin{eqnarray}
R&<& I(V_{11};Y_1|V_{21},X_3,Q)\\
3R&<& I(V_{11},V_{21},X_3;Y_1|Q),
\end{eqnarray}
\end{subequations}
which determines the expression for the maximum symmetric rate given as~\eqref{eq:detcap3}.
\end{enumerate}
\item Converse\\
Converse is proven in similar ways to the previous cases. Let us consider the second term of~\eqref{eq:detcap3}. It can be written as
\begin{subequations}
\begin{eqnarray}
&&3(R_1+R_2+R_3)\nonumber \\
&&\quad < I(V_{11},V_{21},X_3;Y_1|Q)+I(X_1,V_{22},V_{32};Y_2|Q)+I(V_{13},X_2,V_{33};Y_3|Q) \\
&& \quad = H(Y_1|Q)+H(Y_2|Q)+H(Y_3|Q).
\end{eqnarray}
\end{subequations}
Hence, no side information is given to any receiver. Fano's inequality, however, starts from more than three terms which implies appropriate side information for the additional terms need to be given. The appropriate side information is the message from the intended transmitter. From Fano's inequality, we have
\begin{subequations}
\begin{eqnarray}
&&3n(R_1+R_2+R_3)\nonumber \\
&&\quad\leq 3I(X_1^n;Y_1^n)+3I(X_2^n;Y_2^n)+3I(X_3^n;Y_3^n)\\
 &&\quad \leq I(X_1^n;Y_1^n)+2I(X_1^n;Y_1^n,V_{11}^n)+I(X_2^n;Y_2^n)+I(X_2^n;Y_2^n,V_{22}^n)\nonumber\\
 &&\qquad +I(X_3^n;Y_3^n)+I(X_3^n;Y_3^n,V_{33}^n)\\
 &&\quad = H(Y_1^n)-H(V_{21}^n)-H(X_3^n)+2H(V_{11}^n)+H(Y_2^n)-H(V_{32}^n)-H(X_1^n)+2H(V_{22}^n)\nonumber\\
 &&\qquad +H(Y_3^n)-H(V_{13}^n)-H(X_2^n)+2H(V_{33}^n)\\
 &&\quad = H(Y_1^n)+H(Y_2^n)+H(Y_3^n)+2H(V_{11}^n)-H(V_{13}^n)-H(X_1^n)\nonumber\\
 &&\qquad +2H(V_{22}^n)-H(V_{21}^n)-H(X_2^n)+2H(V_{33}^n)-H(V_{32}^n)-H(X_3^n)\\
 &&\quad \leq H(Y_1^n)+H(Y_2^n)+H(Y_3^n)\\
  &&\quad \leq \sum_{i=1}^n \Big (H(Y_{1i})+H(Y_{2i})+H(Y_{3i})\Big)\\
  &&\quad = n\Big(H(Y_1|Q)+H(Y_2|Q)+H(Y_3|Q)\Big).
\end{eqnarray}
\end{subequations}
Remaining bounds can be evaluated similarly.
\end{enumerate}
\end{proof}
\section{Gaussian IC}
\label{sec:gau}
We now consider the GDOF of GIC defined in~\eqref{eq:ch}. To derive the GDOF, we will use $\mathcal{O}(1)$ approximation. We say that $f(x)=g(x)+\mathcal{O}(1)$ when $\lim_{x \rightarrow \infty} |f(x)-g(x)|<\infty$. Note that the GDOF optimality still allows infinite gap to capacity, but $\mathcal{O}(1)$ gap implies the finite gap. Similar to the result in~\cite{GoJa11}, the result obtained in this report is actually stronger than the GDOF because of this $\mathcal{O}(1)$ nature. To analyze behavior of MIMO IC with high SNR, we need the following lemma which is given in~\cite{KaVa11, MoMu11}.
\begin{lemma}
\label{lem:dim}
\cite{MoMu11} Suppose $H_1,H_2,H_3$ are $N \times r$ matrices with rank $r$. When $\alpha>\beta>\gamma$, we have 
\begin{eqnarray}
&&\log |I+\rho^\alpha H_1H_1^H+\rho^\beta H_2H_2^H+\rho^\gamma H_3H_3^H|\nonumber\\
&&\quad = r\alpha\log\rho + \min(r,(N-r)^+)\beta\log\rho + \min(r,(N-2r)^+)\gamma\log\rho + \mathcal{O}(1).
\end{eqnarray}
\end{lemma}
The above lemma essentially says that we can retain full DOF provided by exponents of $\rho$ if there is enough dimension. This is an important property in high SNR regime, and this is why we have~\eqref{eq:self} and~\eqref{eq:int} for the deterministic model. Now we are ready to present the GDOF of GIC. 
\subsection{Case of $\alpha_1<1$}
\begin{theorem}
GDOF of Gaussian IC with $2M\leq N<3M$ and $\alpha_1<1$ is given as 
\begin{eqnarray}
\label{eq:gdof}
&&d_{sym}(\alpha_1,\alpha_2)=\nonumber\\
&&\quad \min \bigg\{ \max\Big\{M+(N-3M)\alpha_2, M+(N-3M)\alpha_1+(3M-N)\alpha_2, (3M-N)\alpha_1+N-2M    \Big\},\nonumber\\
&&\qquad  \qquad \max\Big\{M+\frac{1}{2}(N-3M)\alpha_2, \frac{1}{2}(N-M)+\frac{1}{2}(3M-N)\alpha_2\Big\},\nonumber\\
&&\qquad \qquad M+\frac{1}{3}(N-3M)\alpha_2   \bigg\}.
\end{eqnarray}
\end{theorem}
\begin{proof}
\begin{enumerate}
\item Achievability\\
\begin{enumerate}
\item Codebook generation and encoding\\
The idea essentially the same as the deterministic model. Transmitter $i$ splits its message $W_i$ into $W_{ic_1},W_{ic_2},$\\
$W_{ip}$. $W_{ic_1}$ is encoded using a Gaussian codebook with rate $R_{c_1}$ and covariance $(1-\rho^{-\alpha_2})I$. $W_{ic_2}$ is encoded using a Gaussian codebook with rate $R_{c_2}$ and covariance $(\rho^{-\alpha_2}-\rho^{-\alpha_1})I$. $W_{ip}$ is encoded using a Gaussian codebook with rate $R_{p}$ and covariance $\rho^{-\alpha_1}I$. This power splitting is essentially the same as in~\cite{EtTsWa08, GoJa11}. It can be easily seen that $W_{ip}$ will reach the receiver with the channel strength $\rho^{\alpha_1}$ at the noise level, and $W_{ic_2}+W_{ip}$ will reach the receiver with the channel strength $\rho^{\alpha_2}$ at the noise level. Therefore, each receiver treats those messages as noise. 
\item Decoding\\
Similar to deterministic model, we consider successive decoding. Common messages jointly decoded first while treating other messages as noise, and the private message of the intended transmitter is decoded.  
\item Error analysis\\
We obtain bounds on achievable rate by analyzing error events. It suffices to consider receiver 1 only due to symmetric nature of the channel. First, we have the following relevant single rate bounds.
\begin{subequations}
\begin{eqnarray}
R_p&<&\log \bigg|I+\Big(I+\sum_{i\neq 1}H_{i1}H_{i1}^H\Big)^{-1}\rho^{1-\alpha_1}H_{11}H_{11}^H\bigg|\\
   &=&M(1-\alpha_1)\log \rho +\mathcal{O}(1).
\end{eqnarray}
\end{subequations}
\begin{subequations}
\begin{eqnarray}
R_{c_2}&<&\log \bigg|I+\Big(I+\rho^{1-\alpha_1}H_{11}H_{11}^H+\sum_{i\neq 1}H_{i1}H_{i1}^H\Big)^{-1}(\rho^{1-\alpha_2}-\rho^{1-\alpha_1})H_{11}H_{11}^H\bigg|\\
   &=&M(\alpha_1-\alpha_2)\log \rho +\mathcal{O}(1).
\end{eqnarray}
\end{subequations}
\begin{subequations}
\begin{eqnarray}
R_{c_1}&<&\log \bigg|I+\Big(I+\rho^{1-\alpha_1}H_{11}H_{11}^H+\sum_{i\neq 1}H_{i1}H_{i1}^H\Big)^{-1}(\rho^{\alpha_2}-1)H_{21}H_{21}^H\bigg|\\
   &=&M\alpha_2\log \rho +\mathcal{O}(1).
\end{eqnarray}
\end{subequations}
\begin{subequations}
\begin{eqnarray}
R_{c_2}&<&\log \bigg|I+\Big(I+\rho^{1-\alpha_1}H_{11}H_{11}^H+\sum_{i\neq 1}H_{i1}H_{i1}^H\Big)^{-1}(\rho^{\alpha_1-\alpha_2}-1)H_{31}H_{31}^H\bigg|\\
   &=&M(\alpha_1-\alpha_2)\log \rho +\mathcal{O}(1).
\end{eqnarray}
\end{subequations}
Note that prelog factors of two different bounds on $R_{c_2}$ are the same. This is similar to what happened in deterministic model. The sum rate bound of all common messages is given as 
\begin{subequations}
\begin{eqnarray}
3R_{c_1}+2R_{c_2}&<&\log \bigg|I+\Big(I+\rho^{1-\alpha_1}H_{11}H_{11}^H+\sum_{i\neq 1}H_{i1}H_{i1}^H\Big)^{-1}\nonumber \\
&&\qquad \times \Big((\rho-\rho^{1-\alpha_1})H_{11}H_{11}^H+(\rho^{\alpha_2}-1)H_{21}H_{21}^H+(\rho^{\alpha_1}-1)H_{31}H_{31}^H\Big)\bigg|\\
   &=&(2M\alpha_1+(N-2M)\alpha_2)\log \rho +\mathcal{O}(1).
\end{eqnarray}
\end{subequations}
There are five more bounds on the achievable rate which are relevant to the GDOF, and the prelog factors of those bounds vary depending on actual values of $\alpha_1$ and $\alpha_2$. This is due to the fact that there is not enough dimension to resolve all messages in those bounds. By carefully evaluating it, it turns out that only two of them are relevant eventually, and they as given as follows. 
\begin{eqnarray}
&&R_{c_1}+R_{c_2}< \nonumber \\
&&\begin{cases}
\Big(M\alpha_1+(N-3M)\alpha_2\Big)\log \rho +\mathcal{O}(1), &\text{ if } \alpha_1+\alpha_2<1, 2\alpha_2<\alpha_1\\
\Big((4M-N)\alpha_1+N-3M\Big)\log \rho +\mathcal{O}(1), &\text{ if } \alpha_1+\alpha_2>1, 2\alpha_1-\alpha_2>1\\
\Big((N-2M)\alpha_1+(3M-N)\alpha_2\Big)\log \rho +\mathcal{O}(1), &\text{ if } 2\alpha_1-\alpha_2<1, 2\alpha_2>\alpha_1.
\end{cases}
\end{eqnarray}
\begin{eqnarray}
&&2R_{c_1}+2R_{c_2}< \nonumber \\
&&\begin{cases}
\Big(2M\alpha_1+(N-3M)\alpha_2\Big)\log \rho +\mathcal{O}(1), &\text{ if } \alpha_2<\frac{1}{2}\\
\Big(2M\alpha_1+(3M-N)\alpha_2+N-3M\Big)\log \rho +\mathcal{O}(1), &\text{ if } \alpha_2>\frac{1}{2}.
\end{cases}
\end{eqnarray}
Then, we get the expression for the GDOF given as~\eqref{eq:gdof}.
\end{enumerate}
\item Converse\\
We proceed as deterministic model, i.e., we show each term of~\eqref{eq:gdof} is the upper bound. Let us consider the first term. Let $S^n_{\mathcal{B},i}=\sum_{j \in \mathcal{B}}H_{ji}X^n_{j}+Z^n_i$. We first need to determine side information given to receivers. Note that the first term of~\eqref{eq:gdof} corresponds to the first term of~\eqref{eq:detcap}. Consider now the first term of~\eqref{eq:detcap} as 
\begin{subequations}
\begin{eqnarray}
&& R_1+R_2+R_3 \nonumber\\
&&\quad < I(V_{31},V_{21};Y_1|V_{12},V_{32},Q)+I(X_1;Y_1|V_{12},V_{21},V_{31},Q)\nonumber\\
&&\qquad  +I(V_{12},V_{32};Y_2|V_{13},V_{23},Q)+I(X_2;Y_2|V_{12},V_{23},V_{32},Q)\nonumber\\
&&\qquad  +I(V_{13},V_{23};Y_3|V_{21},V_{31},Q)+I(X_3;Y_3|V_{13},V_{23},V_{31},Q)\\
&&\quad =H(Y_1|V_{12},V_{32},Q)+ H(Y_2|V_{13},V_{23},Q)+H(Y_3|V_{21},V_{31},Q).
\end{eqnarray}
\end{subequations}
This suggests that $V_{12},V_{32}$ need to be given to the receiver 1 as side information, and $V_{13},V_{23}$ to the receiver 2, $V_{21},V_{31}$ to the receiver 3. 
Then, in Gaussian case,
\begin{subequations}
\label{eq:ind}
\begin{eqnarray}
&& n(R_1+R_2+R_3)\nonumber\\
&&\quad \leq I(X_1^n;Y_1^n)+I(X_2^n;Y_2^n)+I(X_3^n;Y_3^n)\\
&&\quad \leq I(X_1^n;Y_1^n,S^n_{\{1,3\},2})+I(X_2^n;Y_2^n,S^n_{\{1,2\},3})+I(X_3^n;Y_3^n,S^n_{\{2,3\},1}).
\end{eqnarray}
\end{subequations}
Consider now $I(X_1^n;Y_1^n,S^n_{\{1,3\},2})$ only. Process for other terms are equivalent to that for $I(X_1^n;Y_1^n,S^n_{\{1,3\},2})$.
\begin{subequations}
\label{eq:up1}
\begin{eqnarray}
&&I(X_1^n;Y_1^n,S^n_{\{1,3\},2})\nonumber\\
&&\quad =h(S^n_{\{1,3\},2})-h(S^n_{3,2})+h(Y_1^n|S^n_{\{1,3\},2})-h(S^n_{\{2,3\},1}|S^n_{3,2})+\mathcal{O}(1)\\
&&\quad =h(Y_1^n|S^n_{\{1,3\},2})+h(S^n_{\{1,3\},2})-h(S^n_{\{2,3\},1},S^n_{3,2})+\mathcal{O}(1)\\
&&\quad = h(Y_1^n|S^n_{\{1,3\},2})+h(S^n_{\{1,3\},2})-h(S^n_{\{2,3\},1})+\mathcal{O}(1).
\end{eqnarray}
\end{subequations}
By evaluating other terms in similar ways, we get 
\begin{subequations}
\begin{eqnarray}
&&n(R_1+R_2+R_3)\nonumber\\
&&\quad \leq h(Y_1^n|S^n_{\{1,3\},2})+h(S^n_{\{1,3\},2})-h(S^n_{\{2,3\},1})+h(Y_2^n|S^n_{\{1,2\},3})+h(S^n_{\{1,2\},3})-h(S^n_{\{1,3\},2})\nonumber\\
&&\qquad +h(Y_3^n|S^n_{\{2,3\},1})+h(S^n_{\{2,3\},1})-h(S^n_{\{1,2\},3})+\mathcal{O}(1)\\
&&\quad = h(Y_1^n|S^n_{\{1,3\},2})+h(Y_2^n|S^n_{\{1,2\},3})+h(Y_3^n|S^n_{\{2,3\},1})+\mathcal{O}(1)\\
&&\quad \leq  n(h(Y_1^G|S^G_{\{1,3\},2})+h(Y_2^G|S^G_{\{1,2\},3})+h(Y_3^G|S^G_{\{2,3\},1}))+\mathcal{O}(1),
\end{eqnarray}
\end{subequations}
where the supersctipt $G$ denotes the inputs are i.i.d. Gaussian with $tr(E[X_iX_i^H])<M$. The last inequality comes from the fact that $h(Y_1^n|S^n_{\{1,3\},2})\leq nh(Y_1^G|S^G_{\{1,3\},2})$ from~\cite{AnVe09}. We have
\begin{eqnarray}
h(Y_1^G|S^G_{\{1,3\},2})=\log \Big |\pi \Sigma_{Y_1^G|S^G_{\{1,3\},2}}\Big |,
\end{eqnarray}
where 
\begin{eqnarray}
\Sigma_{Y_1^G|S^G_{\{1,3\},2}}=E[Y_1^G(Y_1^G)^H]- E[Y_1^G (S^G_{\{1,3\},2})^H    ]E[S^G_{\{1,3\},2}(S^G_{\{1,3\},2})^H]^{-1}E[S^G_{\{1,3\},2}(Y_1^G)^H].
\end{eqnarray}
Therefore, $h(Y_1^G|S^G_{\{1,3\},2})$ can be evaluated by using Woodbury matrix identity, which is 
\begin{eqnarray}
(A+BCD)^{-1}=A^{-1}-A^{-1}B(C^{-1}+DA^{-1}B)^{-1}DA^{-1}. 
\end{eqnarray}
By proceed in similar ways to~\cite{GoJa11}, we eventually get 
\begin{eqnarray}
&&D_{sym}(\alpha_1,\alpha_2)\nonumber\\
&&\quad\leq \max\Big\{M+(N-3M)\alpha_2, M+(N-3M)\alpha_1+(3M-N)\alpha_2, (3M-N)\alpha_1+N-2M    \Big\}.
\end{eqnarray}
To resolve the second term of~\eqref{eq:gdof}, consider the fifth term of~\eqref{eq:detcap} as
\begin{subequations}
\begin{eqnarray}
&& 2(R_1+R_2+R_3) \nonumber\\
&&\quad < I(V_{12},V_{21},V_{31};Y_1|V_{13},Q) +2I(X_1;Y_1|V_{12},V_{21},V_{31},Q)\nonumber\\
&&\qquad  +I(V_{12},V_{23},V_{32};Y_2|V_{21},Q) +2I(X_2;Y_2|V_{12},V_{23},V_{32},Q)\nonumber\\
&&\qquad  +I(V_{13},V_{23},V_{31};Y_3|V_{32},Q) +2I(X_3;Y_3|V_{13},V_{23},V_{31},Q)\\
&&\quad =H(Y_1|V_{13},Q)+H(Y_1|V_{12},V_{21},V_{31},Q)\nonumber\\
&&\qquad +H(Y_2|V_{21},Q)+H(Y_2|V_{12},V_{23},V_{32},Q)\nonumber\\
&&\qquad +H(Y_3|V_{32},Q)+H(Y_3|V_{13},V_{23},V_{31},Q).
\end{eqnarray}
\end{subequations}
In Gaussian case, we proceed as
\begin{subequations}
\begin{eqnarray}
&&2n(R_1+R_2+R_3)\nonumber\\
&&\quad \leq 2I(X_1^n;Y_1^n)+2I(X_2^n;Y_2^n)+2I(X_3^n;Y_3^n)\\
&&\quad \leq I(X_1^n;Y_1^n,S^n_{1,3})+I(X_1^n;Y_1^n,S^n_{1,2},S^n_{\{2,3\},1})\nonumber\\
&&\qquad +I(X_2^n;Y_2^n,S^n_{2,1})+I(X_2^n;Y_2^n,S^n_{2,3},S^n_{\{1,3\},2})\nonumber\\
&&\qquad +I(X_3^n;Y_3^n,S^n_{3,2})+I(X_3^n;Y_3^n,S^n_{3,1},S^n_{\{1,2\},3}) \\
&&\quad = h(S^n_{1,3})+h(Y_1^n|S^n_{1,3})-h(S^n_{\{2,3\},1})+h(S^n_{1,2})+h(Y_1^n|S^n_{1,2},S^n_{\{2,3\},1})\nonumber\\
          &&\qquad +h(S^n_{2,1})+h(Y_2^n|S^n_{2,1})-h(S^n_{\{1,3\},2})+h(S^n_{2,3})+h(Y_2^n|S^n_{2,3},S^n_{\{1,3\},2})\nonumber\\ 
          &&\qquad +h(S^n_{3,2})+h(Y_3^n|S^n_{3,2})-h(S^n_{\{1,2\},3})+h(S^n_{3,1})+h(Y_3^n|S^n_{3,1},S^n_{\{1,2\},3})+\mathcal{O}(1)\\
&&\quad = h(Y_1^n|S^n_{1,3})+h(Y_1^n|S^n_{1,2},S^n_{\{2,3\},1})+h(Y_2^n|S^n_{2,1})+h(Y_2^n|S^n_{2,3},S^n_{\{1,3\},2})\nonumber\\ 
&&\qquad +h(Y_3^n|S^n_{3,2})+h(Y_3^n|S^n_{3,1},S^n_{\{1,2\},3})+ h(S^n_{2,1})+h(S^n_{3,1})-h(S^n_{\{2,3\},1})\nonumber\\ 
&&\qquad + h(S^n_{1,2})+h(S^n_{3,2})-h(S^n_{\{1,3\},2}) + h(S^n_{1,3})+h(S^n_{2,3})-h(S^n_{\{1,3\},2})+\mathcal{O}(1).
\end{eqnarray}
\end{subequations}
Although we cannot usually say that $h(S^n_{2,1})+h(S^n_{3,1})-h(S^n_{\{2,3\},1})\leq n\Big(h(S^G_{2,1})+h(S^G_{3,1})-h(S^G_{\{2,3\},1})\Big)$, this is true if $X_i$ has identity covariance matrix as given in~\cite{Bl65} and discussed in~\cite{AnVe09}. SISO version of $h(S^n_{2,1})+h(S^n_{3,1})-h(S^n_{\{2,3\},1})\leq n\Big(h(S^G_{2,1})+h(S^G_{3,1})-h(S^G_{\{2,3\},1})\Big)$ is given in Lemma 5 of~\cite{AnVe09}, and the corresponding MIMO version which is needed here can be easily obtained by following the exactly same procedure with identity covariance matrix of $X_i$. Since replacing covariance matrix of $X_i$ with an identity matrix only results in $\mathcal{O}(1)$ gap, we can proceed as
\begin{subequations}
\begin{eqnarray}
&&2n(R_1+R_2+R_3)\nonumber\\
&&\quad \leq  n\Big( h(Y_1^G|S^n_{1,3})+h(Y_1^G|S^G_{1,2},S^G_{\{2,3\},1})+h(Y_2^G|S^G_{2,1})+h(Y_2^G|S^G_{2,3},S^G_{\{1,3\},2})\nonumber\\ 
&&\qquad +h(Y_3^G|S^G_{3,2})+h(Y_3^G|S^G_{3,1},S^G_{\{1,2\},3})\Big)+ n\Big(h(S^G_{2,1})+h(S^G_{3,1})-h(S^G_{\{2,3\},1})\nonumber\\ 
&&\qquad + h(S^G_{1,2})+h(S^G_{3,2})-h(S^G_{\{1,3\},2}) + h(S^G_{1,3})+h(S^G_{2,3})-h(S^G_{\{1,3\},2})\Big)+\mathcal{O}(1)\\
&&\quad = n\Big( h(Y_1^G|S^n_{1,3})+h(Y_1^G|S^G_{1,2},S^G_{\{2,3\},1})+h(Y_2^G|S^G_{2,1})+h(Y_2^G|S^G_{2,3},S^G_{\{1,3\},2})\nonumber\\ 
&&\qquad +h(Y_3^G|S^G_{3,2})+h(Y_3^G|S^G_{3,1},S^G_{\{1,2\},3})\Big)+\mathcal{O}(1).
\end{eqnarray}
\end{subequations}
Now we can proceed in similar ways to to the first bound to get 
\begin{eqnarray}
D_{sym}(\alpha_1,\alpha_2)\leq \max\Big\{M+\frac{1}{2}(N-3M)\alpha_2, \frac{1}{2}(N-M)+\frac{1}{2}(3M-N)\alpha_2\Big\}.
\end{eqnarray}
Consider now the third term of~\eqref{eq:gdof}. As seen in the proof of Theorem~\ref{thm:detcap}, the last term of~\eqref{eq:detcap} can be written as
\begin{eqnarray}
&& R_1+R_2+R_3 \nonumber\\
&&\quad =H(Y_1|Q)+H(Y_2|V_{12},V_{21},V_{32},Q)+H(Y_3|V_{13},V_{23},V_{31},Q).
\end{eqnarray}
Given side information to receivers 2 and 3, messages decoded at these receivers are only from the intended transmitters, and these messages do not contain common information to the receiver 1. Hence, we consider a system in which only receiver 1 sees interference. Let $\underline{Y}_2^n=\underline{H}_2\underline{X^n_2}+\underline{Z}_2$, where $\underline{Y}_2^n=[Y_2^n\quad  Y_3^n]^T$, $\underline{X}_2^n=[X_2^n \quad X_3^n]^T$, $\underline{Z}_2^n=[Z_2^n \quad Z_3^n]^T$, and 
\begin{eqnarray}
\underline{H}_2=\begin{bmatrix}
H_{22} & 0\\
0 & H_{33}
\end{bmatrix}.
\end{eqnarray}
If we define $\underline{H}_{21}=[H_{21} \quad H_{31}]$, then we have $ Y_1^n=H_{11}X_1^n+ \underline{H}_{21}\underline{X}_2^n+Z^n_1.$ Let $S^n=\underline{H}_{21}\underline{X}_2^n+Z^n_1.$ An upper bound on achievable rate of this channel is also an upper bound of the channel in~\eqref{eq:ch}. Then,
\begin{subequations}
\label{eq:sum}
\begin{eqnarray}
&&n(R_1+R_2+R_3)\nonumber\\
&&\quad \leq I(X_1^n;Y_1^n)+I(\underline{X}_2^n;\underline{Y}_2^n,S^n )\\
&&\quad =h(Y_1^n)-h(S^n)+h(S^n)+h(\underline{Y}_2^n|S^n)+\mathcal{O}(1)\\
&&\quad =h(Y_1^n)+h(\underline{Y}_2^n|S^n)+\mathcal{O}(1).
\end{eqnarray}
\end{subequations}
Now we can proceed in similar ways to the previous cases to get 
\begin{eqnarray}
D_{sym}(\alpha_1,\alpha_2)\leq M+\frac{1}{3}(N-3M)\alpha_2.
\end{eqnarray}
\end{enumerate}
\end{proof}
It can be seen that evaluation of $I(X_1^n;Y_1^n,S^n_{\{1,3\},2})$ in~\eqref{eq:up1} is more complicated than its deterministic counter part. If we have an assumption of~\eqref{eq:intdec} as in the deterministic model, then we would have $h(S^n_{\{1,3\},2})-h(S^n_{3,2})=h(S^n_{1,2})$. Although this becomes eventually true when we replace everything with Gaussian, we cannot assume such property at that point. This makes analysis on Gaussian model significantly more involved than the deterministic model. \\
\indent 
In~\cite{GoJa11}, two active outer bounds are many-to-one bound and bound on all common messages from unintended transmitters. These are counterparts of two outer bounds in~\cite{EtTsWa08} for symmetric case. Note that a many-to-one bound is still an active outer bound in \textit{partially asymmetric} model, but other two bounds in this model do not have clear counterparts in aforementioned bounds. The second outer bound, however, has similar way of derivation to the bound on all common messages from unintended transmitters.~\cite[Lemma 5]{GoJa11}. Note that an important property used in~\cite{GoJa11} is symmetry, e.g., replacing $S^n_{1,2}$ with $S^n_{1,3}$ helps resolving terms. In \textit{partially asymmetric} model, however, it does not help, and hence, vector entropy power inequality~\cite{Bl65} is needed to further approximate the upper bound. \\
\indent
Let us now further evaluate GDOF expression in~\eqref{eq:gdof}. By careful evaluation, we get
\begin{eqnarray}
\label{eq:gdofsim}
&&\hspace{-20pt}D_{sym}(\alpha_1,\alpha_2)=\nonumber\\
&&\hspace{-20pt}\begin{cases}
M+(N-3M)\alpha_2 &\hspace{-10pt}\text{ if } \alpha_1+\alpha_2<1, 2\alpha_2<\alpha_1\\
\min \Big\{M+ (N-3M)\alpha_1 +(3M-N)\alpha_2, M+ \frac{1}{2}(N-3M)\alpha_2 \Big\}&\hspace{-10pt}\text{ if } 2\alpha_1-\alpha_2<1, 2\alpha_2>\alpha_1, \alpha_2<\frac{1}{2}\\
\min \Big\{N-2M+ (3M-N)\alpha_1, M+ \frac{1}{2}(N-3M)\alpha_2 \Big\}&\hspace{-20pt}\text{ if } \alpha_1+\alpha_2>1, 2\alpha_1-\alpha_2>1, \alpha_2<\frac{1}{2}\\
\min \Big\{\frac{1}{2}(N-M)+ \frac{1}{2}(3M-N)\alpha_2, M+ \frac{1}{3}(N-3M)\alpha_2 \Big\}&\hspace{-10pt}\text{ if } \alpha_1+\alpha_2>1, \alpha_2>\frac{1}{2}.
\end{cases}
\end{eqnarray}
Note that the first term $M+(N-3M)\alpha_2$ is the DOF which can be obtained by treating interference as noise (IAN). In 3-user symmetric case considered in~\cite{GoJa11}, the optimal GDOF is strictly larger than the GDOF obtained by IAN for all interference regimes while 2-user SISO case in~\cite{EtTsWa08} also has interference regime in which IAN is GDOF optimal. It is argued in~\cite{GoJa11} that the reason why there is no interference regime in which IAN is optimal in 3-user symmetric case is because there are multiple receiver antennas. In SISO case, there is no other dimension to resolve interference common information when interference strength is weak enough to be affected by private message strength of the intended user. In SIMO case, however, there always are other dimensions to exploit when interference strength is weak. Then why do we see different behavior in asymmetric case? If we look at the condition for which IAN is GDOF optimal, then we can see that $\alpha_2$ must be small with $\alpha_1$ being at least as twice as larger than $\alpha_2$, and both $\alpha_1$ and $\alpha_2$ cannot be too large. Therefore, we may think that this comes from the difficulty of decoding $W_{c_1}$ from weaker interference channel due to strong interference from stronger channel and the private message of the intended transmitter. If $\alpha_1$ is large enough, then the power of the private message of the intended transmitter becomes small enough. The condition $\alpha_1+\alpha_2<1$, however, prevents it which would result in no common message through weaker interference link. When signal from weaker interference link is treated as noise, the receiver has enough dimension to resolve all remaining messages, and hence IAN is optimal.\\
\indent
Another observation can be made with the bounds corresponding to the first and the third term in~\eqref{eq:gdof}. From~\eqref{eq:gdofsim}, we can see that the third bound which is called many-to-one bound is only active when $\alpha_1+\alpha_2>1, \alpha_2>\frac{1}{2}$ which corresponds to stronger interference. If we look at side information given in~\eqref{eq:sum}, then we can see that the actual constraint comes from decodability of all common messages at receiver 1. If we look at side information given in~\eqref{eq:ind}, then the actual constraint does not depend on decodability of common messages of the intended transmitter. In other words, limiting factor in this case is the decodability of common messages from unintended transmitters given correct decoding of common messages from the intended transmitter. Hence, this bound must be active for weaker interference while many-to-one bound is active for stronger interference. 
\subsection{Case of $\alpha_2<1<\alpha_1$}
\begin{theorem}
\label{thm:gdof2}
GDOF of Gaussian IC with $2M\leq N<3M$ and $\alpha_2<1<\alpha_1$ is given as 
\begin{eqnarray}
\label{eq:gdof2}
&&d_{sym}(\alpha_1,\alpha_2)=\nonumber\\
&&\quad \min \bigg\{ M,  \frac{2M+M\alpha_1+(N-3M)\alpha_2}{3},\nonumber\\
&&\qquad  \qquad \max\Big\{ \frac{M+M\alpha_1+(N-3M)\alpha_2}{2}, \frac{M+(N-2M)\alpha_1+(3M-N)\alpha_2}{2}            \Big\} \bigg\}.
\end{eqnarray}
\end{theorem}
\begin{proof}
\begin{enumerate}
\item Achievability\\
\begin{enumerate}
\item Codebook generation and encoding\\
Transmitter $i$ splits its message $W_i$ into $W_{ic_1},W_{ic_2}$. $W_{ic_1}$ is encoded using a Gaussian codebook with rate $R_{c_1}$ and covariance $(1-\rho^{-\alpha_2})I$. $W_{ic_2}$ is encoded using a Gaussian codebook with rate $R_{c_2}$ and covariance $\rho^{-\alpha_2}I$. 
\item Decoding\\
Similar to deterministic model, all messages are decoded jointly. 
\item Error analysis\\
We obtain bounds on achievable rate by analyzing error events. By proceeding similarly to the case of $\alpha_1<1$, we get the following bounds which are relevant for GDOF. 
\begin{eqnarray}
R_{c_2}&<&M(1-\alpha_2)\log \rho +\mathcal{O}(1).
\end{eqnarray}
\begin{eqnarray}
R_{c_1}&<&M\alpha_2\log \rho +\mathcal{O}(1).
\end{eqnarray}
\begin{eqnarray}
&&2R_{c_1}+2R_{c_2}< \nonumber \\
&&\begin{cases}
\Big(M+M\alpha_1+(N-3M)\alpha_2 \Big)\log \rho +\mathcal{O}(1), &\text{ if } \alpha_2<\frac{\alpha_1}{2}\\
\Big(M+(N-2M)\alpha_1+(3M-N)\alpha_2 \Big)\log \rho +\mathcal{O}(1), &\text{ if } \alpha_2>\frac{\alpha_1}{2}.
\end{cases}
\end{eqnarray}
\begin{eqnarray}
3R_{c_1}+2R_{c_2}&<&  \Big(M+M\alpha_1+(N-2M)\alpha_2\Big)\log \rho +\mathcal{O}(1).
\end{eqnarray}
Then, we get the expression for the GDOF given as~\eqref{eq:gdof2}.
\end{enumerate}
\item Converse\\
We show each term of~\eqref{eq:gdof2} is the upper bound. The first term of~\eqref{eq:gdof2} is trivially an upper bound. Let us consider the third term. Let $S^n_{\mathcal{B},i}=\sum_{j \in \mathcal{B}}H_{ji}X^n_{j}+Z^n_i$. From the proof of Theorem~\ref{thm:detcap2}, side information to each receiver can be determined. 
Then, in Gaussian case,
\begin{subequations}
\begin{eqnarray}
&& 2n(R_1+R_2+R_3)\nonumber\\
&&\quad \leq 2I(X_1^n;Y_1^n)+2I(X_2^n;Y_2^n)+2I(X_3^n;Y_3^n)\\
&&\quad \leq I(X_1^n;Y_1^n,S^n_{3,2})+I(X_1^n;Y_1^n,S^n_{1,2})\nonumber\\
&&\qquad + I(X_2^n;Y_2^n,S^n_{1,3})+I(X_2^n;Y_2^n,S^n_{2,3})\nonumber\\
&&\qquad + I(X_3^n;Y_3^n,S^n_{2,1})+I(X_3^n;Y_3^n,S^n_{3,1}).
\end{eqnarray}
\end{subequations}
Consider now $I(X_1^n;Y_1^n,S^n_{3,2})+I(X_1^n;Y_1^n,S^n_{1,2})$ only.
\begin{subequations}
\begin{eqnarray}
&&I(X_1^n;Y_1^n,S^n_{3,2})+I(X_1^n;Y_1^n,S^n_{1,2})\nonumber\\
&&\quad =h(Y_1^n|S^n_{3,2})-h(S^n_{\{2,3\},1}|S^n_{3,2})+h(S^n_{1,2})\nonumber\\
&&\qquad +h(Y_1^n|S^n_{1,2})-h(S^n_{\{2,3\},1})+\mathcal{O}(1)\\
&&\quad \leq h(Y_1^n|S^n_{3,2})-h(S^n_{\{2,3\},1}|S^n_{3,2})-h(S^n_{3,2})+h(S^n_{3,2})+h(S^n_{1,2})\nonumber\\
&&\qquad + h(S^n_{\{2,3\},1})+h(S^n_{1,1}|S^n_{1,2})-h(S^n_{\{2,3\},1})+\mathcal{O}(1)\\
&&\quad \leq h(Y_1^n|S^n_{3,2})-h(S^n_{\{2,3\},1})+h(S^n_{3,2})+h(S^n_{1,2})+h(S^n_{1,1}|S^n_{1,2})+\mathcal{O}(1).
\end{eqnarray}
\end{subequations}
By evaluating other terms in similar ways, we get 
\begin{subequations}
\begin{eqnarray}
&&2n(R_1+R_2+R_3)\nonumber\\
&&\quad \leq h(Y_1^n|S^n_{3,2})-h(S^n_{\{2,3\},1})+h(S^n_{3,2})+h(S^n_{1,2})+h(S^n_{1,1}|S^n_{1,2})\nonumber\\
&&\qquad + h(Y_2^n|S^n_{1,3})-h(S^n_{\{1,3\},2})+h(S^n_{1,3})+h(S^n_{2,3})+h(S^n_{2,2}|S^n_{2,3})\nonumber\\
&&\qquad + h(Y_3^n|S^n_{2,1})-h(S^n_{\{1,2\},3})+h(S^n_{2,1})+h(S^n_{3,1})+h(S^n_{3,3}|S^n_{3,1})\\
&&\quad = h(Y_1^n|S^n_{3,2})+ h(Y_2^n|S^n_{1,3})+ h(Y_3^n|S^n_{2,1})\nonumber\\
&&\qquad +h(S^n_{1,1}|S^n_{1,2})+h(S^n_{2,2}|S^n_{2,3})+h(S^n_{3,3}|S^n_{3,1})\nonumber\\
&&\qquad +h(S^n_{2,1})+h(S^n_{3,1})-h(S^n_{\{2,3\},1})\nonumber\\
&&\qquad +h(S^n_{1,2})+h(S^n_{3,2})-h(S^n_{\{1,3\},2})\nonumber\\
&&\qquad +h(S^n_{1,3})+h(S^n_{2,3})-h(S^n_{\{1,2\},3})+\mathcal{O}(1)\\
&&\quad \leq  nh(Y_1^G|S^G_{3,2})+ nh(Y_2^G|S^G_{1,3})+ nh(Y_3^G|S^G_{2,1})\nonumber\\
&&\qquad +nh(S^G_{1,1}|S^G_{1,2})+nh(S^G_{2,2}|S^G_{2,3})+nh(S^G_{3,3}|S^G_{3,1})\nonumber\\
&&\qquad +nh(S^G_{2,1})+nh(S^G_{3,1})-nh(S^G_{\{2,3\},1})\nonumber\\
&&\qquad +nh(S^G_{1,2})+nh(S^G_{3,2})-nh(S^G_{\{1,3\},2})\nonumber\\
&&\qquad +nh(S^G_{1,3})+nh(S^G_{2,3})-nh(S^G_{\{1,2\},3})+\mathcal{O}(1)\\
&&\quad = n(h(Y_1^G|S^G_{3,2})+ h(Y_2^G|S^G_{1,3})+ h(Y_3^G|S^G_{2,1}))+\mathcal{O}(1).
\end{eqnarray}
\end{subequations}
We eventually get 
\begin{eqnarray}
&&D_{sym}(\alpha_1,\alpha_2)\nonumber\\
&&\quad\leq \max\Big\{ \frac{M+M\alpha_1+(N-3M)\alpha_2}{2}, \frac{M+(N-2M)\alpha_1+(3M-N)\alpha_2}{2}            \Big\}.
\end{eqnarray}
The second term of~\eqref{eq:gdof2} can be evaluated similarly.
\end{enumerate}
\end{proof}
Note that the second term of~\eqref{eq:gdof2} comes from many-to-one bound, and hence converse for this bound would be proven in similar ways to the case of $\alpha_1<1$. This also implies that many-to-one bound of the case of $\alpha_1<1$ would be proven in more systematic, albeit less intuitive way given in the proof of Thereom~\ref{thm:gdof2}.\\
\indent 
It can be seen from~\eqref{eq:gdof2} that the GDOF can be $M$ which implies that the effect of interference is completely removed. In symmetric cases in~\cite{EtTsWa08, GoJa11, JaVi10}, the effect of interference is removed when interference is much stronger than the desired channel. Theorem~\ref{thm:gdof2} shows that it can happen even when some of interferences are weaker than the desired channel for asymmetric case.
\subsection{Case of $\alpha_2>1$}
\begin{theorem}
\label{thm:gdof3}
GDOF of Gaussian IC with $2M\leq N<3M$ and $\alpha_2>1$ is given as 
\begin{eqnarray}
\label{eq:gdof3}
&&d_{sym}(\alpha_1,\alpha_2)= \min \{ M,  \frac{N-2M+M\alpha_1+M\alpha_2}{3} \}.
\end{eqnarray}
\end{theorem}
\begin{proof}
\begin{enumerate}
\item Achievability\\
\begin{enumerate}
\item Codebook generation and encoding\\
Transmitter $i$'s message $W_{i}$ is encoded using a Gaussian codebook with rate $R$ and covariance $I$.
\item Decoding\\
Similar to deterministic model, all messages are decoded jointly. 
\item Error analysis\\
We obtain bounds on achievable rate by analyzing error events. By proceeding similarly to the previous cases, we get the following bounds which are relevant for GDOF. 
\begin{eqnarray}
R&<&M\log \rho +\mathcal{O}(1).
\end{eqnarray}
\begin{eqnarray}
3R&<&  \Big(N-2M+M\alpha_1+M\alpha_2\Big)\log \rho +\mathcal{O}(1).
\end{eqnarray}
Then, we get the expression for the GDOF given as~\eqref{eq:gdof3}.
\end{enumerate}
\item Converse\\
We show each term of~\eqref{eq:gdof3} is the upper bound. The first term of~\eqref{eq:gdof3} is trivially an upper bound. Let us consider the second term. Let $S^n_{\mathcal{B},i}=\sum_{j \in \mathcal{B}}H_{ji}X^n_{j}+Z^n_i$. From the proof of Theorem~\ref{thm:detcap3}, side information to each receiver can be determined. 
Then, in Gaussian case,
\begin{subequations}
\begin{eqnarray}
&& 3n(R_1+R_2+R_3)\nonumber\\
&&\quad \leq 3I(X_1^n;Y_1^n)+3I(X_2^n;Y_2^n)+3I(X_3^n;Y_3^n)\\
&&\quad \leq I(X_1^n;Y_1^n)+2I(X_1^n;Y_1^n,S^n_{1,1})\nonumber\\
&&\qquad + I(X_2^n;Y_2^n)+2I(X_2^n;Y_2^n,S^n_{2,2})\nonumber\\
&&\qquad + I(X_3^n;Y_3^n)+2I(X_3^n;Y_3^n,S^n_{3,3}).
\end{eqnarray}
\end{subequations}
Consider now $I(X_1^n;Y_1^n)+2I(X_1^n;Y_1^n,S^n_{1,1})$ only.
\begin{subequations}
\begin{eqnarray}
&&I(X_1^n;Y_1^n)+2I(X_1^n;Y_1^n,S^n_{1,1})\nonumber\\
&&\quad =h(Y_1^n)-h(S^n_{\{2,3\},1})+2h(S^n_{1,1})+\mathcal{O}(1)\\
&&\quad =h(Y_1^n)+h(S^n_{2,1})+h(S^n_{3,1})-h(S^n_{\{2,3\},1})+2h(S^n_{1,1})-h(S^n_{2,1})-h(S^n_{3,1})+\mathcal{O}(1).
\end{eqnarray}
\end{subequations}
By evaluating other terms in similar ways, we get 
\begin{subequations}
\begin{eqnarray}
&&3n(R_1+R_2+R_3)\nonumber\\
&&\quad \leq h(Y_1^n)+h(Y_2^n)+h(Y_3^n)+h(S^n_{2,1})+h(S^n_{3,1})-h(S^n_{\{2,3\},1})\nonumber\\
&&\qquad +h(S^n_{1,2})+h(S^n_{3,2})-h(S^n_{\{1,3\},2})+h(S^n_{1,3})+h(S^n_{2,3})-h(S^n_{\{1,2\},3})\nonumber\\
&&\qquad +2h(S^n_{1,1})-h(S^n_{1,2})-h(S^n_{1,3})+2h(S^n_{2,2})-h(S^n_{2,1})-h(S^n_{2,3})\nonumber\\
&&\qquad +2h(S^n_{3,3})-h(S^n_{3,1})-h(S^n_{3,2})+\mathcal{O}(1)\\
&&\quad \leq h(Y_1^n)+h(Y_2^n)+h(Y_3^n)+h(S^n_{2,1})+h(S^n_{3,1})-h(S^n_{\{2,3\},1})\nonumber\\
&&\qquad +h(S^n_{1,2})+h(S^n_{3,2})-h(S^n_{\{1,3\},2})+h(S^n_{1,3})+h(S^n_{2,3})-h(S^n_{\{1,2\},3})+\mathcal{O}(1)\\
&&\quad \leq nh(Y_1^G)+nh(Y_2^G)+nh(Y_3^G)+nh(S^G_{2,1})+nh(S^G_{3,1})-nh(S^G_{\{2,3\},1})\nonumber\\
&&\qquad +nh(S^G_{1,2})+nh(S^G_{3,2})-nh(S^G_{\{1,3\},2})+nh(S^G_{1,3})+nh(S^G_{2,3})-nh(S^G_{\{1,2\},3})+\mathcal{O}(1)\\
&&\quad = n(h(Y_1^G)+ h(Y_2^G)+ h(Y_3^G))+\mathcal{O}(1).
\end{eqnarray}
\end{subequations}
We eventually get 
\begin{eqnarray}
&&D_{sym}(\alpha_1,\alpha_2) \leq \frac{N-2M+M\alpha_1+M\alpha_2}{3}.
\end{eqnarray}
\end{enumerate}
\end{proof}
From~\eqref{eq:gdof3}, we can directly conclude that the effect of interference is removed, i.e, the GDOF is $M$, if $\alpha_1+\alpha_2>\frac{5M-N}{M}$ which corresponds to very strong interference. It can be also seen that setting $\alpha_1=\alpha_2=\alpha$ exactly recovers the result in the symmetric case in~\cite{GoJa11} for $\alpha>1$.
\subsection{Further interpretation of the GDOF}
Figure~\ref{fig:gdof} describes the GDOF region of 3-user MIMO GIC with $M=1$ and $N=2$ for $0<\alpha_1,\alpha_2<2$. It can be seen that the GDOF region consists of several faces corresponding to terms in~\eqref{eq:gdof},~\eqref{eq:gdof2},~\eqref{eq:gdof3}. 
\begin{figure}[ht]
\begin{center}{
 \includegraphics[width=0.9\textwidth]{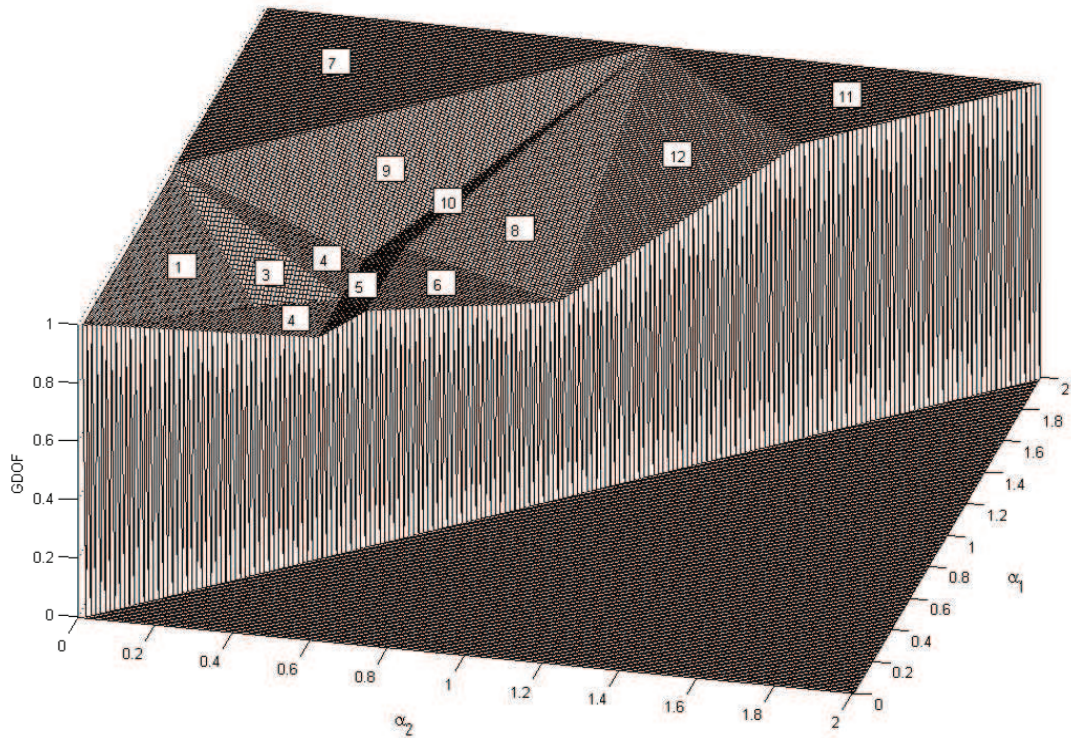}  
 }\end{center}
  \caption{The GDOF region of 3-user $1 \times 2$ GIC}
  \label{fig:gdof}
\end{figure}
For ease of exposition, each face is labeled in Figure~\ref{fig:gdof} by the order of appearance of the corresponding term in~\eqref{eq:gdof},~\eqref{eq:gdof2},~\eqref{eq:gdof3}. For example,~\eqref{eq:gdof} has six terms in total, and faces corresponding to these terms are labeled from 1 to 6. \eqref{eq:gdof2} has four terms, and faces corresponding to these terms are labeled from 7 to 10. \eqref{eq:gdof3} has two terms, and faces corresponding to these terms are labeled from 11 to 12. Let the face with label $i$ be called the face $i$. The face 2 intersects with the faces 1,3, and 4, and it does not appear in Figure~\ref{fig:gdof} due to the viewing angle. Note that three triangular faces of 1,3 and 4 form a pit in the GDOF region.\\
\indent
The face 7 in Figure~\ref{fig:gdof} represents the GDOF of $M$. In this region, the message from the stronger interference can be decoded by treating everything else as noise. After decoding of the stronger interference, the receiver has enough dimension to resolve everything else. The face 11 also represents the GDOF of $M$. In this region, the messages from two interferences are jointly decoded by treating the desired signal as noise.  
\indent
In symmetric case, faces 4,6,11 and 12 are only active, and the intersection of them with the plane $\alpha_1=\alpha_2$ forms a line corresponds to the GDOF of 3-user symmetric SIMO GIC shown in~\cite[Figure 2]{GoJa11}. As mentioned earlier, the face 1 corresponds to the GDOF of IAN. An orthogonal scheme achieves the GDOF of $2/3$, and it is a lower bound of the GDOF region in Figure~\ref{fig:gdof} with equality at the points corresponding to the intersection of the faces 6,8, and 12 and the intersection of the faces 1,2 and 3.\\
\indent
Let us consider now weak interference regime which corresponds to faces 1 to 6 more carefully. There are three important factors in understanding the GDOF behavior. Faces 1 to 6 correspond to the rate bounds given in~\eqref{eq:gdof} at the receiver 1 as mentioned before. The first important factor is the rate terms involved in each of these bounds. The second important factor comes from Lemma~\ref{lem:dim}. From this lemma, it can be seen that the user whose messages have the weakest power cannot achieve the full DOF when receiver dimension is not enough. Depending on values of $\alpha_1$ and $\alpha_2$, the power of each user's messages changes as well as the user with the weakest power. This is the second important factor. Table~\ref{tab:weak} describes these factors for bounds correspond to faces 1 to 6 at the receiver 1. 
\begin{table}[!t]\footnotesize
\renewcommand{\arraystretch}{1.3}
\caption{Table of important factors in each bound at the receiver 1}
\label{tab:weak}
\centering
\begin{tabular}{|c|c|c|}
 \hline
   \textbf{Bound} & \textbf{Involved rate terms} & \textbf{The smallest power exponent and the corresponding user}  \\
  \hline
  \hline
  1 & \multirow{3}{*}{$R_{1p}, R_{2c_1}, R_{3c_2}$  } & $\alpha_2$ (user 2)\\
  \cline{1-1}\cline{3-3}
  2 && $\alpha_1-\alpha_2$ (user 3)\\
  \cline{1-1}\cline{3-3}
  3 &&  $1-\alpha_1$ (user 1)\\
  \hline
  4& \multirow{2}{*}{$R_{1p}, R_{1c_2}, R_{2c_1}, R_{3c_1}, R_{3c_2}$  } & $\alpha_2$ (user 2)\\
  \cline{1-1}\cline{3-3}
  5&& $1-\alpha_2$ (user 1)\\
  \hline
  6& $R_{1p}, R_{1c_1}, R_{1c_2}, R_{2c_1}, R_{3c_1}, R_{3c_2}$ & $\alpha_2$ (user 2)\\
  \hline 
\end{tabular}
\end{table}
The third important factor is the signal strength of each user's messages $W_{c_1}, W_{c_2}, W_{p}$ at the receiver 1, and they are described in Table~\ref{tab:str}.
\begin{table}[!t]\footnotesize
\renewcommand{\arraystretch}{1.3}
\caption{the signal strength of each user's messages at the receiver 1}
\label{tab:str}
\centering
\begin{tabular}{|c||c|c|c|}
 \hline
   \textbf{ \backslashbox{Message}{User}   } & \textbf{1} & \textbf{2} & \textbf{3} \\
  \hline
  \hline
  $R_p$ & $\rho^{1-\alpha_1}$ & \multirow{2}{*}{1} & 1\\
  \cline{1-2}\cline{4-4}
  $R_{c_2}$ & $\rho^{1-\alpha_2}-\rho^{1-\alpha_1}$ & & $\rho^{\alpha_1-\alpha_2}-1$\\
  \hline
  $R_{c_1}$ & $\rho-\rho^{1-\alpha_2}$ & $\rho^{\alpha_2}-1$ & $\rho^{\alpha_1}-\rho^{\alpha_1-\alpha_2}$\\
  \hline 
\end{tabular}
\end{table}
From Figure~\ref{fig:gdof}, it can be seen that bound 6 is active when $\alpha_2$ is larger, and bounds 4 and 5 is active when $\alpha_2$ is smaller. This can be explained by universal common message $W_{c_1}$. From Table~\ref{tab:weak}, it can be seen that bound 6 includes $R_{1c_1}$ while bounds 4 and 5 do not. When $\alpha_2$ is small, the decodability of $W_{2c_1}$ and $W_{3c_1}$ becomes limiting factor as suggested in Table~\ref{tab:str}. As $\alpha_2$ increases, this decodability increases such that decodability of $W_{1c_1}$ gets important. This is similar to phenomena observed in symmetric cases in~\cite{EtTsWa08, GoJa11}. Note that bounds 1 to 3 occurs only if $\alpha_2<0.5$. Figure~\ref{fig:weakgdof} describes the GDOF with $\alpha_2=0.2$ with labels which match with corresponding faces in Figure~\ref{fig:gdof}. 
\begin{figure}[ht]
\begin{center}{
 \includegraphics[width=0.5\textwidth]{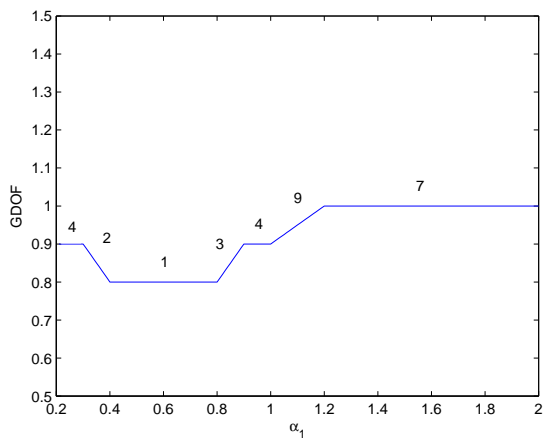}  
 }\end{center}
  \caption{The GDOF of 3-user $1 \times 2$ GIC with $\alpha_2=0.2$}
  \label{fig:weakgdof}
\end{figure}
It can be seen that bound 4 is first active until $\alpha_1$ reaches certain level. This can be again explained by $W_{c_1}$. From Table~\ref{tab:weak}, it can be seen that bounds 1 to 3 do not involve $R_{3c_1}$ while bound 4 does. As suggested in Table~\ref{tab:str}, increasing $\alpha_1$ increases decodability of $W_{3c_1}$, and decodability of only $W_{2c_1}$ becomes limiting factor at some point. Interestingly, bound 4 is active again when $\alpha_1$ is larger. This can be explained by $W_{c_2}$. From Table~\ref{tab:weak}, it can be seen that bounds 1 to 3 do not involve $R_{1c_2}$ while bound 4 does. As suggested in Table~\ref{tab:str}, increasing $\alpha_1$ increases decodability of $W_{3c_2}$, and decodability of $W_{1c_2}$ becomes limiting factor together with $W_{3c_2}$ at some point.\\
\indent
Let us consider bound 6. Increasing $\alpha_2$ with fixing $\alpha_1$ decreases the GDOF. As seen in Table~\ref{tab:weak}, the weakest user in this bound is user 2 with power exponent $\alpha_2$. Increasing $\alpha_2$ increases decodability of $W_{2c1}$ while decreasing decodability of $W_{1c2}$ and $W_{3c2}$. Since $W_{2c1}$ does not achieve the full DOF, rate gain of $W_{2c1}$ cannot compensate rate loss of $W_{1c2}$ and $W_{3c2}$. On the other hand, Increasing $\alpha_1$ with fixing $\alpha_2$ does not change the GDOF. In this case, changing $\alpha_1$ does not affect the weakest user's power, and hence rate gain and loss occur for users which achieve full DOF. This maintains the balance. \\
\indent
Opposite phenomenon occurs in bound 5. Increasing $\alpha_2$ with fixing $\alpha_1$ increases the GDOF. Increasing $\alpha_2$ increases decodability of $W_{2c1}$ while decreasing decodability of $W_{1c2}$ and $W_{3c2}$. As seen in Table~\ref{tab:weak}, the weakest user in this bound is user 1 with power exponent $1-\alpha_2$. Since user 1 does not achieve the full DOF, rate gain of user 2 exceeds rate loss of user 1. Increasing $\alpha_1$ with fixing $\alpha_2$ does not change the GDOF from the similar reason to the case of bound 5.\\
\indent
Things get a little trickier when $\alpha_2<0.5$. Consider Figure~\ref{fig:weakgdof}. On bound 4, the weakest user is user 2 with power exponent $\alpha_2$, and hence increasing $\alpha_1$ maintains balance. On bound 2, however, the weakest user is user 3 with power exponent $\alpha_1-\alpha_2$. Hence, increasing $\alpha_1$ increases signal power of the user who does not achieve the full DOF, and this results in decrease of the GDOF. Once it hits bound 1 which corresponds to IAN, increasing $\alpha_1$ does not affect the GDOF. On bound 3, the weakest user is user 1 with power exponent $1-\alpha_1$. Increasing $\alpha_1$ decreases decodability of user 1 who does not achieve the full DOF while increasing decodability of other user's messages. This should result in increase of the GDOF.
\section{Conclusion}
\label{sec:con}
In this paper, the GDOF of 3-user MIMO GIC is characterized. As conjectured in~\cite{GoJa11}, Han-Kobayashi or Etkin-Tse-Wang-like message splitting achieves the GDOF although generalization of multiple message splitting is required. Three messages per transmitter suffice in 3-user case, and it implies that $K$ message splitting would achieve the GDOF of any $K$-user MIMO GIC which satisfies the condition of $M(K-1)\leq N$. Note that the GDOF result obtained in this paper essentially implies $\mathcal{O}(1)$ gap to the capacity which is finite, but the exact gap to the capacity cannot be computed. In finite SNR regime in which the exact constant gap is desirable, degraded nature does not exist in the channel anymore, and this means that more message splitting would be required to establish such result. \\
\indent
We define \textit{partially asymmetric} GIC which yields valuable insights with manageable amount of computation on asymmetric GIC with more than two users which has not been well studied in literature. The form of the GDOF region gives interesting interpretation on interactions among different interferences for the first time in information theoretic study on interference channel. The methodology in this paper would achieve the GDOF of any MIMO GIC which satisfies $M(K-1)\leq N$ probably through cumbersome analysis. \\
\indent
The deterministic model is used in this paper as in~\cite{GoJa11} to facilitate easier analysis for the Gaussian case. The most important benefit is the systematic way of determining side information for converse. In Gaussian case, however, the proof of converse is not as simple as in the deterministic case due to the fact that the channel output becomes linear combination of channel inputs. Asymmetric nature of the channel requires an approximation which implies non-trivial generalization of the symmetric case. This approximation using vector entropy inequality can possibly be used for general $K$-user MIMO GIC, but it relies on the fact that the GDOF is not affected by finite gap to the capacity, and hence, may not be used for analysis in finite SNR regime. This implies that a better upper bounding technique is required to obtain the exact constant cap to the capacity.


\bibliographystyle{IEEEtran}
\bibliography{bae}
\end{document}